\documentclass[a4paper,USenglish]{oasics-v2018}

\usepackage{microtype}

\title{Multi-Source Multi-Sink Nash Flows Over Time\footnote{This research was carried out in the framework of \textsc{Matheon} supported by Einstein Foundation Berlin.}}

\titlerunning{Multi-Source Multi-Sink Nash Flows Over Time}

\author{Leon Sering}{Institute of Mathematics, Technische Universität Berlin\\Straße des 17. Juni 136, 10623 Berlin, Germany}{sering@math.tu-berlin.de}{}{}

\author{Martin Skutella}{Institute of Mathematics, Technische Universität Berlin\\Straße des 17. Juni 136, 10623 Berlin, Germany}{martin.skutella@tu-berlin.de}{}{}

\authorrunning{L. Sering and M. Skutella}

\Copyright{L. Sering and M. Skutella}

\subjclass{;
}


\category{}

\supplement{}

\funding{}

\acknowledgements{The authors are much indebted to Roberto Cominetti and Jos\'e Correa for interesting discussions and for sharing their insights and thoughts on the topic of this paper.}

\hideOASIcs
\nolinenumbers

\usepackage{mathtools}

\usepackage{braket}

\makeatletter
\newcommand*{\inlineequation}[2][]{%
  \begingroup
    \refstepcounter{equation}%
    \ifx\\#1\\%
    \else
      \label{#1}%
    \fi
    \relpenalty=10000 %
    \binoppenalty=10000 %
    \ensuremath{%
      #2%
    }%
    ~\@eqnnum
  \endgroup
}
\makeatother

\renewcommand\labelenumi{(\roman{enumi})}
\renewcommand\theenumi\labelenumi

\definecolor{myColor}{RGB}{0,0,0}

\newcommand{\R}{\mathbb{R}}
\newcommand{\N}{\mathbb{N}}

\newcommand*\diff{\mathop{}\!\mathrm{d}}
\newcommand{\abs}[1]{\left\lvert #1 \right\rvert}

\newcommand{\Sources}{S^+}
\newcommand{\Sinks}{S^-}
\newcommand{\source}{s}
\newcommand{\sink}{t}
\newcommand{\T}{T}
\newcommand{\Flow}{\mathbb{R}_{\geq 0}}

\newcommand{\exG}{\bar G}
\newcommand{\exE}{\bar E}
\newcommand{\exV}{\bar V}
\newcommand{\exf}{\bar f}

\graphicspath{{figures/}}

\newcommand{\Paths}{\mathcal{P}}

\renewcommand{\l}{\ell}
\newcommand{\n}{n}
\newcommand{\m}{m}

\begin{document}
  \maketitle
  
\begin{abstract}
Nash flows over time describe the behavior of selfish users eager to reach their
destination as early as possible while traveling along the arcs of a network with
capacities and transit times. Throughout the past decade, they have been thoroughly
studied in single-source single-sink networks for the deterministic queuing model,
which is of particular relevance and frequently used in the context of traffic and
transport networks. In this setting there exist Nash flows over time that can be
described by a sequence of static flows featuring special properties, so-called
`thin flows with resetting'. This insight can also be used algorithmically to
compute Nash flows over time.
We present an extension of these results to networks with multiple sources and sinks
which are much more relevant in practical applications. In particular, we come up with a
subtle generalization of thin flows with resetting, which yields a compact description as
well as an algorithmic approach for computing multi-terminal Nash flows over time.
\end{abstract}

\section{Introduction} 
With the emergence of novel navigation and vehicle technologies (including, e.g.,
self-driving/smart vehicles) along with the availability of massive amounts of data in
todays and future traffic and transportation networks, increasing attention is given to
the mathematical modeling and algorithmic solution of the interplay of individual agents
in such networks. We study the behavior of selfish users who wish to travel through a
traffic or transportation network. While there is already a vast amount of literature
and results on steady states of such systems (see, e.g.,
Roughgarden~\cite{Roughgarden05} and the references therein), much less is known about
the often more realistic but also much more complex situation of such systems evolving
and changing over time.

\vspace{2cm}
\subparagraph{Flows over time.}
Flows over time provide an excellent mathematical model for agents (flow particles)
traveling through a network over time, with capacities and transit times (delays) on the
arcs. Flows over time have been introduced in a seminal paper by Ford and
Fulkerson~\cite{FordFulkerson58} and can also be found in their classic
textbook~\cite{FordFulkerson62}. For a given single-source single-sink network with
capacities and transit times on the arcs and a given time horizon, they show how to
efficiently construct a maximum flow over time, that is, a way of sending as much flow
as possible from the source to the sink within the given time horizon. The underlying
algorithm is based on a static min-cost flow computation in the given network where arc
transit times are interpreted as costs. A decomposition of the static flow into flows
along source-sink-paths then provides an optimal strategy for sending flow over time
from the source to the sink by using each path as long as possible.

Surprisingly, and in contrast to the situation known for classic (i.e., static) network
flows, the problem of balancing given supplies and demands in a network with several
sources and/or sinks by sending flow within a given time horizon turns out to be
considerably more difficult and complicated. Following the work of Ford and Fulkerson,
it took almost four decades before Hoppe and Tardos~\cite{HoppeTardos00} came up with an
efficient algorithm for solving this transshipment over time problem; see also Hoppe's
PhD thesis~\cite{Hoppe95}. Their algorithm, however, while being theoretically
efficient, relies on parametric submodular function minimization, leading to unpleasant
and usually unrealistic running times for networks of practical sizes. Only recently,
Schlöter and Skutella~\cite{SchloeterSkutella2017} presented a slight improvement of
this result. Another somewhat surprising evidence for the increased difficulty of flow
over time problems compared to static flow problems is the fact that the computation of
(fractional) multicommodity flows over time constitutes an NP-hard
problem~\cite{HallHipplerSkut-TCS}. We refer to~\cite{Skutella_2009} for a recent survey
on and thorough introduction to flows over time.

\subparagraph{Nash equilibria for the deterministic queuing model.} 
The flow over time problems discussed in the previous paragraph are all based on the
assumption that flow particles are controlled by a central authority who decides the
route choices and schedules of the particles. In most realistic traffic situations,
however, the lack of coordination among flow particles necessitates an additional game
theoretic perspective. We assume that each flow particle is an individual agent that
seeks to arrive at a destination in the least possible time. Such models have mostly
been studied in the transportation literature; see, e.g., the book by Ran and
Boyce~\cite{RanBoyce96} for an overview.

In this paper we study Nash equilibria for flows over time in the deterministic queuing
model that is also at the core of many large-scale agent-based traffic simulations such as,
e.g., MATSim; see~\cite{Horni2016}. Here the actual transit time of a flow particle along
an arc is the sum of the arc's free-flow transit time plus the waiting time spent in a
queue that builds up whenever more flow tries to use an arc than the arc's capacity can
handle. In particular, the first-in-first-out (FIFO) principle holds. 
We refer to Section~\ref{sec:flow_dynamics} for a detailed definition.

For a single-source single-sink network, Koch and Skutella~\cite{Koch_2009} characterize
Nash flows over time featuring a special and very useful structure: Their derivatives
are piece-wise constant, therefore constituting a sequence of particular static
source-sink flows, so-called \emph{thin flows with resetting}. Exploiting this key concept of
thin flows with resetting, Cominetti, Correa, and Larré~\cite{Cominetti_2015} provide a
constructive proof for the existence and uniqueness of equilibria in this setting, using
a fixed-point formulation. Furthermore, for the more general case of multiple
origin-destination pairs, they provide a non-constructive existence proof. For the
single-source single-sink setting, Cominetti, Correa, and Olver~\cite{Cominetti2017}
show that, for networks with sufficient capacity, a dynamic equilibrium reaches a
steady state in finite time.

\subparagraph{Our contribution.} Our structural and algorithmic understanding of Nash flows
over time is limited to the very restrictive special case of single-source single-sink
networks. Moreover, in contrast to the classical case of static flows, single-commodity
flows over time in multi-source multi-sink networks with given supplies and demands cannot
easily be reduced by introducing a super-source and a super-sink; see, e.g., the work of
Hoppe and Tardos~\cite{HoppeTardos00} discussed above. Nevertheless, we show that such a
reduction is possible, albeit non-trivial, when considering a particularly meaningful model
of Nash flows over time in such networks. This leads to an interesting generalization of
the structural and algorithmic results known for the single-source single-sink case;
see~\cite{Koch_2009,Cominetti_2015,Cominetti2017}. In particular, we present an appropriate
generalization of `thin flows with resetting' and prove that a Nash flow over time can be
described and algorithmically obtained via a sequence of these static flows. As another
interesting aspect of this work, we show how to get rid of the identification of flow
particles with the time they enter the network which has been used in previous work on
the single-source single-sink case. In our more general model, all flow is
waiting in front of the sources of the network right from the beginning, a subtle point
that turns out to be crucial for being able to handle multiple source nodes.

\subparagraph{Outline.}

In Section~\ref{sec:settings} we informally describe several different settings for dynamic
routing games with multiple sources and sinks and identify a suitable model for our
purposes. Section~\ref{sec:flow_dynamics} introduces the necessary concepts and notations
for describing Nash flows over time. Then, Section~\ref{sec:dynamic_routing_game} explains
how to deal with multiple source nodes. Finally, in Section~\ref{sec:demands} multiple
sinks are considered as well. 

\section{Settings for Routing Games with Multiple Sources and Sinks.}
\label{sec:settings} 

There are several different settings for dynamic routing games when considering multiple
sources and multiple sinks. We discuss the most meaningful interpretations in
the following.
  
Nash flows over time are mainly motivated by dynamic traffic assignments which naturally
lead to the consideration of multiple commodities with independent
origin-destination-pairs~$(\source_i, \sink_i)$ and inflow rates $r_i \geq 0$, for~$i =
1, \dots, \n$. At each origin~$\source_i$, a flow enters the network with rate~$r_i$ and
every infinitesimal small particle of this flow has the goal to reach
destination~$\sink_i$ as early as possible while considering all other particles from
the past and the future. For every commodity, there are time dependent in- and outflow
rates for every arc that must satisfy flow conservation at every node. A dynamic
equilibrium then consists of a flow over time with $n$ commodities, where each particle
chooses a combination of fastest routes from $\source_i$ to~$\sink_i$ as strategy. Note
that queues build up on arcs whenever the inflow rate exceeds the arc's capacity. This
causes a delay of all subsequent particles, therefore influencing the traversing time
of all routes using this arc. Cominetti et~al.~\cite{Cominetti_2015} prove that these
dynamic equilibria exist by using variational inequalities for the path-based
formulation. Unfortunately, the known techniques for single commodity flows are not
sufficient for analyzing or algorithmically constructing such dynamic multi-commodity
Nash flows over time. The fact that each commodity has different earliest arrival times
at the nodes is the main difficulty as this causes cyclic interdependencies between the
commodities. Each particle entering the network has to take into account not only all
flow that previously entered the network, but also flow entering the network
subsequently; an illustrative example is given in the left part of
Figure~\ref{fig:introduction}.

\begin{figure}[h]
\centering
\begin{subfigure}{0.48\textwidth}
\centering
\includegraphics{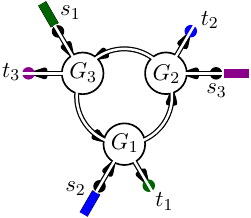} 
\end{subfigure}
\quad
\begin{subfigure}{0.48\textwidth}
\centering 
\includegraphics{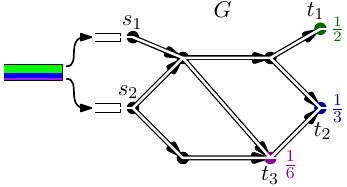} 
\end{subfigure}
\caption{\emph{Left:} Illustration of cyclic interdependencies of commodities with different
origin-destination-pairs. The waiting times for flow from $\source_1$ to $\sink_1$
within subnetwork~$G_1$ depend on flow starting later from $\source_2$ to~$\sink_2$. The
cyclic symmetry implies that a particle has to take into account not only previous but
also future flow from all sources. \emph{Right:} An example of the setting considered in this article. Each flow particle may
choose whether to enter the network at $\source_1$ or~$\source_2$, but the flow is
partitioned according to the demands at the sinks. Here one half of the flow has~$\sink_1$
as its destination, one third wants to reach~$\sink_2$, and the rest of one sixth aims
at~$\sink_3$.}
\label{fig:introduction}
\end{figure}%
When we relax the pairing of origins and destinations, however, the route choice of each
particle only depends on flow that previously entered the network. We stick to
individual inflow rates for the sources, but instead of matching the sources to
destinations, we consider $m$ sinks~$\sink_1, \dots, \sink_{\m}$ with demands~$d_1,
\dots, d_{\m} \geq 0$, such that $d_1 + \dots + d_{\m} = 1$. The value~$d_j$ denotes the
share of the total flow entering the network that has $\sink_j$ as destination. In
terms of traffic networks this means that each road user has a predetermined destination,
but may choose between multiple origins to enter the network; see right side of
Figure~\ref{fig:introduction}. In order to obtain well defined Nash flows
over time with unique arrival times we exclude situations as described in
Figure~\ref{fig:no_queues_at_sources} by considering queues in front of the sources. 
\begin{figure}[t]
\centering
\begin{subfigure}{0.3\textwidth}
\centering 
\includegraphics[page=1]{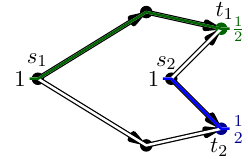}
\caption{} 
\label{fig:no_queues_at_sources_1}
\end{subfigure}
\quad
\begin{subfigure}{0.3\textwidth}
\centering
\includegraphics[page=2]{no_queues_at_sources}
\subcaption{}
\label{fig:no_queues_at_sources_2}
\end{subfigure}
\begin{subfigure}{0.3\textwidth}
\centering
\includegraphics[page=3]{no_queues_at_sources}
\subcaption{}
\label{fig:no_queues_at_sources_3}
\end{subfigure}
\caption{Assume that, at each point in time, the total inflow is equally divided
according to the demands. Then, in this symmetric instance, a possible dynamic
equilibrium sends all flow from $\source_1$ to $\sink_1$ and from $\source_2$ to
$\sink_2$~(a). Alternatively, the destinations might be
swapped~(b). Equilibrium (a), however, heavily benefits sink~$\sink_2$ by serving it
earlier as the path from $\source_1$ to~$\sink_1$ is longer than the path
from~$\source_2$ to~$\sink_2$. Symmetrically, equilibrium (b) benefits~$\sink_1$.
This is the reason for considering queues in front of the sources such that a Nash flow over
time is forced to behave symmetrically in this instance, that is, sinks with the
same demand are treated equally in every equilibrium; see~(c).}
\label{fig:no_queues_at_sources}
\end{figure}  %
In other words, there is essentially one flow~$\Flow$ consisting of a continuum
of infinitesimally small particles~$\phi\in\Flow$, where each splittable particle
chooses, in the order given by $<$, a convex combination of fastest routes from the sources to the sinks as strategy. The sum of the coefficients of all paths to
sink~$\sink_j$ has to be equal to demand~$d_j$. Each particle is then split according to these
coefficients and each part is sent along its route. How these choices of strategies
can be constructed, and what structure these Nash flows over time have, is discussed in
this paper.

In the case of one source and multiple sinks with given
demands, the two settings presented above are equivalent: given a
multi-origin-destination instance with one source but $n$ commodities, we can construct
an equivalent multi-source multi-sink instance by setting the inflow at the source to $r
\coloneqq r_1 + \dots + r_{\n}$ and the demand at sink~$\sink_j$ to $d_j \coloneqq r_j /
r$. It is easy to see that these settings are also equivalent in the case of
multiple sources and one sink.

\section{Flow Dynamics}
\label{sec:flow_dynamics}

In this section we present all necessary definitions of a \emph{fluid queuing network}. The model is a modified version
used by Koch and Skutella \cite{Koch_2009} and Cominetti et al.\ \cite{Cominetti_2011, Cominetti_2015,Cominetti2017}
that matches the multi-source multi-sink setting.
  
  Throughout this paper we consider a directed graph $G = (V, E)$ with transit times~$\tau_e \geq 0$ and capacities~$\nu_e > 0$
  on every arc~$e \in E$, a set of $\n \geq 1$ sources $\Sources = \Set{\source_1, \dots, \source_{\n}} \subseteq V$ with inflow rates~$r_1, \dots, r_{\n} > 0$, and a set of $\m \geq 1$ sinks~$\Sinks = \Set{\sink_1, \dots, \sink_\m} \subseteq V$. The corresponding
  demands will be introduced in Section \ref{sec:demands}. We assume that every node is reachable by a source and can itself
  reach at least one sink. Furthermore, we assume that the sum of transit times along every directed cycle is positive.
  
  \subparagraph*{Flows over time} 
  A flow over time is specified by locally integrable and bounded functions~$f^+_e: [0,
  \infty) \to [0, \infty)$ for every arc~$e$. These \emph{inflow functions} describe the rate of flow entering
  the arcs for every point in time~$\theta \in [0, \infty)$. We set $f^+_e(\theta) \coloneqq 0$ for $\theta < 0$.
  
  \begin{figure}[t]
    \centering
    \includegraphics{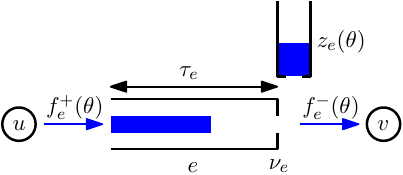} 
    \caption{A snapshot of arc $e = uv$ at time $\theta$, with transit time $\tau_e$, capacity $\nu_e$, inflowrate $f_e^+(\theta)$, outflowrate $f_e^-(\theta)$, and queue $z_e(\theta)$.}
    \label{fig:edge_dynamics}
  \end{figure} 
  
  For every arc~$e$ there is a bottleneck given by its capacity~$\nu_e$ at the head of the
  arc.\footnote{The dynamics are exactly the same if the bottleneck is located at the tail of the arc or anywhere
  between tail and head.} When flow enters $e$ it immediately starts to traverse this arc, which takes
  $\tau_e$ time. If the rate of flow trying to leave $e$ exceeds the capacity~$\nu_e$, the flow builds up a queue in
  front of the bottleneck which is described by a function~$z_e\colon [0,\infty) \to [0, \infty)$. Note that the
  queue does not have any physical dimension in the network, and is therefore called \emph{point queue}. Whenever there is a positive queue the outflow rate  operates at capacity rate $\nu_e$. This leads to the following evolution of the queue starting with $z(0) = 0$,
  \begin{equation}\label{eqn:z'_definition}
  z'_e(\theta) \coloneqq \begin{cases}
  f_e^+(\theta-\tau_e) - \nu_e & \text{ if } z_e(\theta) > 0 \\
  \max \Set{f_e^+(\theta-\tau_e) - \nu_e, 0} & \text{ if } z_e(\theta) = 0.
  \end{cases}
  \end{equation}
  This determines a unique queue function $z_e$ \cite{Cominetti_2015}, which is characterized later on.
  The outflow rate function~$f^-_e\colon [0, \infty) \to [0,
  \infty)$ is defined by
  \begin{equation}\label{eqn:defi_outflow}
  f_e^-(\theta) \coloneqq \begin{cases}
  \nu_e & \text{ if } z_e(\theta) > 0, \\
  \min\Set{f_e^+(\theta - \tau_e), \nu_e} & \text{ if } z_e(\theta) = 0.
  \end{cases}
  \end{equation}
  
    A \emph{flow over time} is given by a family of inflow functions $(f^+_e)_{e\in E}$ that \emph{conserve flow} at every~$v \in V\backslash \Sinks$, which means that the following equation holds for almost all $\theta \in [0, \infty)$:
  \begin{equation} \label{eqn:flow_conservation}
  \sum_{e\in \delta^+(v)} f_e^+(\theta) - \sum_{e \in \delta^-(v)} f_e^-(\theta) = \begin{cases}
  0 & \text{ if } v \in V\setminus \Sources, \\
  r_i & \text{ if } v = \source_i \in \Sources.
  \end{cases}\end{equation}
  This ensures that the network does not leak at intermediate vertices and that the amount of flow entering through
  source~$\source_i$ matches the inflow rate~$r_i$.
    
  The \emph{cumulative in- and outflow} of an arc $e$ is the total amount of flow that has entered or left $e$ up to some point in time $\theta$ and is defined by
  $F^+_e(\theta) \coloneqq \int_{0}^{\theta} f_e^+(\xi) \diff \xi$ and $F^-_e(\theta) \coloneqq \int_{0}^{\theta} f_e^-(\xi) \diff \xi$.
  The amount of flow in the queue of an arc $e$ at time $\theta$ equals the difference between the amount of flow that has entered the queue before time $\theta$ and the flow that has left the queue up to this point in time. The former can be described by the amount of flow that has entered arc $e$ at time $\theta - \tau_e$. In short,
  $z_e(\theta) = F_e^+(\theta - \tau_e) - F_e^-(\theta)$.
  \textcolor{myColor}{[See Lemma~\ref{lemma:in_out_flow} in \ref{ap:lemma_in_out_flow}.]} Since $f^+_e$ and $f^-_e$ are bounded, the
  functions $F^+_e$, $F^-_e$, and $z_e$ are Lipschitz continuous, and therefore almost everywhere differentiable due to Rademacher's theorem~\cite{Rademacher1919}. Considering that $f_e^+(\theta)$ and $f_e^-(\theta)$ are non-negative and $z'_e(\theta) \geq - \nu_e$ for all $\theta$ it follows that $F^+_e$ and $F^-_e$ are non-decreasing and $z_e$ cannot decrease faster than with slope~$-\nu_e$.
  
  We identify the flow with the non-negative reals~$\Flow$, that is, each~$\phi\in\Flow$ corresponds to an infinitesimally small
  flow particle. The natural ordering~$\leq$ corresponds to the priority among the
  flow particles when entering the network, i.e., particle $\phi$ has priority over all $\phi' > \phi$. Consequently, all flow that wants to enter the network through the same source does this in order of priority. 
  Note that the flow represented by the non-negative reals has a width of~$1$. That is, there is exactly one unit of
  flow associated with every unit interval~$[a, a + 1] \subseteq \Flow$.
  To distinguish between flow and time we write $\Flow$ for the ordered set of flow particles,
  mostly denoted by $\phi$ or $\varphi$, and $[0, \infty)$ for the
  time whose elements are \emph{points in time}, often denoted by~$\theta$ or $\vartheta$.
  
  A family of locally integrable functions $f_i\colon \Flow \to [0, 1]$, for $i =
  1, \dots, \n$, is called \emph{inflow distribution} if $\sum_{i = 1}^\n f_i(\phi) = 1$ for almost all $\phi \in
  \Flow$ and if each \emph{cumulative source inflow} 
  $F_i(\phi) \coloneqq \int_{0}^{\phi} f_i(\varphi) \diff \varphi$
  is unbounded for~$\phi \to \infty$.   
  The function $f_i(\phi)$ describes the fraction of particle $\phi$ that enters the network trough $s_i$. The
  cumulative source inflow functions have to be unbounded in order to guarantee that the inflow rates at the sources
  never run dry.

  \subparagraph*{Current shortest paths network}   
  Given a flow over time $(f^+_e)_{e\in E}$ together with an inflow distribution $(f_i)_{i
  = 1}^{\n}$, the \emph{arc travel time} for arc $e$ is the function $\T_e\colon  [0, \infty) \to [0, \infty)$  that
  maps the entrance time $\theta$ to the exit time~$\T_e(\theta)$. More precisely, if a particle enters $e$ at time
  $\theta$, it traverses the arc first, which takes $\tau_e$ time, and then queues up and has to wait in line for
  $z_e(\theta+\tau_e) / \nu_e$ time units. Hence,
  $\T_e(\theta) \coloneqq \theta + \tau_e + z_e(\theta+\tau_e)/\nu_e$.
  We
  require the flow to satisfy the \emph{first in first out (FIFO) condition} on every arc, that is, no particle can overtake other flow on an arc or
  in a queue. Suppose flow particle $\phi$ enters $e$ at time~$\theta$, then the amount of flow which has entered $e$
  before $\phi$ is exactly the amount of flow that leaves $e$ before time $\T_e(\theta)$ when $\phi$
  leaves the arc. In short $F_e^+(\theta) = F_e^-(\T_e(\theta))$.
  \textcolor{myColor}{[Lemma~\ref{lemma:in_out_flow} in \ref{ap:lemma_in_out_flow}.]}

  For every $i = 1, \dots, \n$, the \emph{source arrival time function} maps each particle $\phi \in \Flow$ to the time it arrives at $\source_i$ and is given by $\T_i(\phi) \coloneqq F_i(\phi) / r_i$.  
  Given an $\source_i$-$v$ path $P = (e_1, e_2, \dots, e_k)$ the \emph{arrival time function}~$\T_P\colon \Flow\to[0,\infty)$ maps the particle~$\phi$ to the time at which $\phi$ arrives at $v$ if it traverses the path $P$, i.e.,
  $\T_P(\phi) \coloneqq \T_{e_k} \circ \T_{e_{k-1}} \circ \dots \circ \T_{e_1} \circ \T_i (\phi)$.  
  Since the functions $z_e$ and $F_i$ are Lipschitz continuous, the same holds for $\T_e$,$\T_i$, and $\T_P$. Note that the queue length $z_e$ cannot decrease faster than with slope $-\nu_e$ and, therefore all these
  $\T$-functions are nondecreasing. Furthermore, all $\T$-functions go to infinity for $\phi \to
  \infty$ since the queue lengths $z_e$ are non-negative and the $F_i$ are unbounded.
  
  The \emph{earliest arrival time function} $\l_v\colon  \Flow \to [0, \infty)$ of
  node $v \in V$ maps each particle~$\phi$ to the earliest time~$\l_v(\phi)$ it can possibly
  reach node~$v$. We have
    $\l_v(\phi) = \min_{P \in \Paths_v} \T_P(\phi)$, where $\Paths_v$ is the set of all paths from some source $\source \in \Sources$ to $v$.
  Note that these node-labels are also Lipschitz continuous, nondecreasing, and unbounded and that they are the unique solutions to the following Bellman equations:
    \begin{align}  \label{eqn:bellman}
    \begin{aligned}
    \l_{\source_i}(\phi) &= \min\left(\{\:\T_i(\phi)\:\} \cup \Set{\T_e(\l_u(\phi)) | e = u \source_i\in E}\right)
    &&\quad \text{ for } i = 1, \dots, \n, \\ \l_v(\phi) &= \min_{e = u v\in E} \T_e(\l_u(\phi))  &&\quad \text{ for } v
    \in V\backslash \Sources.
    \end{aligned}
    \end{align}
    This is well defined since all cycles in $G$ have by assumption positive travel times.
  
  For a fixed particle~$\phi$ we call an arc $e = uv$ \emph{active for~$\phi$} if $\l_v(\phi) = \T_e(\l_u(\phi))$ holds. With $E'_\phi$ we denote the set of all active arcs for
  particle~$\phi$ and the subgraph $G'_\phi = (V, E'_\phi)$ is called the \emph{current shortest paths network}. 
  Note that the current shortest paths network is always acyclic since the sum of transit times of each directed cycle is positive.
  
\section{Multi-Source Single-Sink Nash Flows over Time}
\label{sec:dynamic_routing_game} 
  For this section we only consider fluid queuing networks with exactly one sink~$\sink$.
  A flow over time together with an inflow distribution corresponds to a strategy profile, where the strategy of each particle consists of a convex combination of $\Sources$-$\sink$-paths. The following definition characterizes a Nash equilibrium.
  
  \begin{definition}[Nash flow over time] \label{defi:Nash_flow} 
   A tuple $f = ((f^+_e)_{e\in E},(f_i)_{i = 1}^\n)$ consisting of a flow over time and an inflow distribution  is a
   \emph{Nash flow over time}, also called \emph{dynamic equilibrium}, if the following two \emph{Nash flow conditions} hold:
   \begin{align}
   \l_{\source_i}(\phi) &= \T_i(\phi)  &&\text{ for all } i = 1, \dots,\n \text{ and almost all } \phi \in \Flow, \label{eqn:Nash_condition_source} \tag{N1}\\
   f^+_e(\theta) &> 0 \;\;\Rightarrow\;\; \theta \in \l_u(\Phi_e) &&\text{ for all arcs } e = uv \in E \text{ and almost all } \theta \in [0, \infty),\tag{N2}\label{eqn:Nash_condition_active} 
   \end{align}
     where $\Phi_e \coloneqq \set{\phi \in \Flow | e \in E'_\phi}$ is the set of flow particles for which arc $e$ is
     active.
\end{definition}
  
  Figuratively speaking, these two conditions mean, that entering the network through a source $s_i$ is always a fastest
  way to reach $s_i$ (\ref{eqn:Nash_condition_source}) and that a Nash flow over time uses only active arcs (\ref{eqn:Nash_condition_active}), and therefore only shortest paths to $t$. More precisely, particle
  $\phi$ reaches $\sink$ at time $\l_{\sink}(\phi)$ by using active arcs only, and $\l_{\sink}(\phi)$ is the earliest
  time $\phi$ can possibly reach $\sink$ under the assumption that the routes of all previous
  particles~$\varphi<\phi$ are fixed. Since this is true for all particles, a Nash flow over time is indeed a Nash equilibrium.
  
  \begin{lemma} \label{lemma:nash_flow_characterization}
  A tuple~$f = ((f^+_e)_{e\in E}, (f_i)_{i = 1}^\n)$ of a flow over time and an inflow distribution is a Nash flow over
  time if, and only if, we have $F_e^+(\l_u(\phi)) = F_e^-(\l_v(\phi))$ and $F_i(\phi) = \l_{\source_i}(\phi) \cdot r_i$
  for all arcs $e = uv \in E$, every $i = 1, \dots, \n$, and all particles $\phi \in \Flow$.
  \end{lemma}
  
  Lemma~\ref{lemma:nash_flow_characterization} \textcolor{myColor}{[proven in \ref{ap:lemma_nash_flow_characterization}]} motivates to consider the \emph{underlying static flow} for every particle~$\phi$, which is defined by
  $x_e(\phi) \coloneqq F_e^+(\l_u(\phi)) =  F_e^-(\l_v(\phi))$ and
  $x_i(\phi) \coloneqq F_i(\phi) = \l_{\source_i} (\phi) \cdot r_i$.
  For a fixed $\phi$ this is indeed a static $\Sources$-$\sink$-flow since the integral of \eqref{eqn:flow_conservation} over $[0, \l_v(\phi)]$ yields
  \begin{equation} \label{eqn:x_is_flow}
  \sum_{e\in \delta^+(v)} x_e(\phi) - \sum_{e \in \delta^-(v)} x_e(\phi) = \begin{cases}
    0 & \text{ if } v \in V \backslash (\Sources \cup \set{\sink}), \\ 
    \l_{\source_i} (\phi) \cdot r_i = x_i(\phi) & \text{ if } v = \source_i \in \Sources.
    \end{cases}\end{equation}
 Let $x'_e$, $x'_i$, and $\l'_v$ denote the derivative functions, which exist almost everywhere, since the $x$- and $\l$-functions are Lipschitz continuous. It is possible to determine the inflow function of every arc $e = uv$ from these derivatives, since
  $x'_e(\phi) = f_e^+(\l_u(\phi)) \cdot \l'_u(\phi)$.
  Moreover, the inflow distribution is given by
  $f_i(\phi) = \l'_{\source_i}(\phi) \cdot r_i$.
  Consequently, a Nash flow over time is completely characterized by these derivatives. 
  Note that differentiating~\eqref{eqn:x_is_flow} yields that $x'(\phi)$ also forms a static $\Sources$-$\sink$-flow, which has very specific properties that are characterized in the following.

  \subparagraph*{Thin flows with resetting for multiple sources and a single sink} 
  
  A thin flow with resetting is a static flow defined on a subgraph of $G$ characterizing the strategies of particles in a flow interval of a Nash flow over time. The definition of thin flows with resetting given in this article
  generalizes the thin flows with resetting introduced in~\cite{Koch_2009} and the normalized thin flows with resetting
  from~\cite{Cominetti_2015}, in order to suit the multi-source setting.
  
  Let $E' \subseteq E$ be a subset of arcs such that the subgraph $G' = (V, E')$ is acyclic and every node is reachable by
  a source within $G'$. Note that not every node needs to be able to reach sink~$\sink$. Additionally, we
  consider a subset of arcs $E^* \subseteq E'$, called \emph{resetting arcs}. Moreover, let $K(E', x'_1, \dots,
  x'_{\n})$ be the set of all static $\Sources$-$\sink$-flows in~$G'$ with inflow~$x'_i$ at source~$\source_i$ for~$x'_i \geq 0$ and $x'_1 + \dots + x'_{\n} = 1$.
  \begin{definition}[Thin flow with resetting] \label{defi:thin_flow}
    A vector $(x'_i)_{i = 1}^\n$, with $x'_i \geq 0$ and $x'_1 + \dots + x'_{\n} = 1$, and a static flow $(x'_e)_{e \in E}
    \in K(E', x'_1, \dots, x'_{\n})$ together with a node labeling $(\l'_v)_{v \in V}$ is called \emph{thin flow with
    resetting} on $E^* \subseteq E'$ if:
    \begin{alignat}{2}
    \l'_{\source_i} &= x'_i / r_i &&\text{ for all } i = 1, \dots, \n, \label{eqn:l'_s} \tag{TF1}\\ 
    \l'_{\source_i} &\leq \min_{e = u\source_i \in E'} \rho_e(\l'_u, x'_e) \quad && \text{ for all } i = 1, \dots, \n,
    \label{eqn:l'_s_min}\tag{TF2}\\ 
    \l'_v &= \min_{e = uv \in E'} \rho_e(\l'_u, x'_e) \quad && \text{ for all } v \in V \backslash \Sources,
    \label{eqn:l'_v_min}\tag{TF3}\\ 
    \l'_v &= \rho_e(\l'_u, x'_e) && \text{ for all } e = uv \in E' \text{ with } x'_e > 0, \label{eqn:l'_v_tight}\tag{TF4}
    \end{alignat}
    \[\text{ where } \qquad \rho_e(\l'_u, x'_e) \coloneqq \begin{cases}
    x'_e / \nu_e & \text{ if } e = uv \in E^*,\\
    \max\Set{\l'_u, x'_e / \nu_e} & \text{ if } e = uv \in E'\backslash E^*.  
    \end{cases}\]
  \end{definition}
  The next theorem states that the derivatives of a Nash flow over time $f$ form almost everywhere a
  thin flow with resetting on the arcs with positive queues. Recall that $E'_\phi$ is the subset of arcs that are
  active for $\phi$, and let 
  $E^*_\phi \coloneqq \Set{e = uv \in E | z_e(\l_u(\phi) + \tau_e) > 0}$ 
  be the set of arcs where the particle $\phi$ would experience a queue.
  \begin{theorem} \label{thm:l'_equations}
    For a Nash flow over time $((f^+_e)_{e\in E}, (f_i)_{i = 1}^\n)$, the derivative labels $(x'_i(\phi))_{i =
    1}^\n$ and $(x'_e(\phi))_{e\in E'_\phi}$ together with $(\l'_v(\phi))_{v\in V}$ form a thin flow with resetting on
    $E^*_\phi$ in the current shortest paths network $G'_\phi = (V, E'_\phi)$, for almost all~$\phi \in \Flow$.
  \end{theorem}
  The intuitive idea is that $x'_e / \nu_e$ describes the \emph{congestion} of arc $e$ and $\rho_e(\l'_u, x'_e)$ is the \emph{congestion} of all paths to $v$ using $e$. The higher this congestion is, the longer it will take for following particles to reach $v$, which is captured by a high derivative of the earliest arrival time $\l'_v$.   
  If we have $\l'_v < \rho_e(\l'_u, x'_e)$ this means that $e$ leaves the current shortest paths network, and therefore it cannot be used by following particles, i.e., $x'_e = 0$.
  \textcolor{myColor}{[Detailed proof in \ref{ap:thm_l'_equations}.]}
  
  The reverse of Theorem~\ref{thm:l'_equations} is also true in the sense that we can use thin flows with resetting to
  construct a Nash flow over time. For this we first show that there always exists a thin flow with resetting for any
  acyclic graph and any subset of resetting arcs.
  \begin{theorem} \label{thm:existence_of_NTF}
    Consider an acyclic graph~$G' = (V, E')$ with sources~$\Sources$, sink~$\sink$, capacities~$\nu_e$, and a subset of arcs~$E^* \subseteq E'$ and suppose every node is reachable by a source. Then there exists a thin flow $\left( (x'_i)_{i = 1}^\n, (x'_e)_{e \in E}, (\l'_v)_{v\in
    V}\right)$ with resetting on~$E^*$.
  \end{theorem}
  The proof is essentially given in~\cite{Cominetti_2015} and is only slightly modified to fit the new definition of a
  thin flow with resetting. The key idea is to use a set-valued function in order to apply the Kakutan'i fixed-point
  theorem. \textcolor{myColor}{[Detailed proof in \ref{ap:thm_existence_of_NTF}.]}
  
  \subparagraph*{Constructing Nash flows}
  Note that in a dynamic equilibrium no particle can overtake any other particle, and therefore the choice of strategy for $\phi$ only depends on the strategies of the particles in $[0, \phi)$. So we may assume that the particles decide in order of priority.
  More precisely, given a Nash flow over time up to some $\phi \in \Flow$,
  it is possible to extend it by using a thin flow on the $G'_\phi$ with resetting on $E^*_\phi$.
  
  A \emph{restricted Nash flow
  over time} on $[0,\phi]$ is a Nash flow over time where only the particles in $[0, \phi]$ are considered, i.e., for $i
  = 1,\dots, n$ we have $f_i(\varphi) = 0$ for all $\varphi > \phi$ and for each arc~$e = uv \in E$ we have
  $f_e(\l_u(\theta)) = 0$ for all $\theta > \l_u(\phi)$. But the Nash flow conditions
  (\ref{eqn:Nash_condition_source}) and (\ref{eqn:Nash_condition_active}) are satisfied for almost all particles in $[0,
  \phi]$ and almost all times in $[0, \l_u(\phi)]$.
  
  Since all previous results carry over to restricted Nash flows over time, the earliest arrival times $(\l_v)_{v\in V}$ are well-defined for particles in $[0, \phi]$, and therefore it is possible to determine~$G'_\phi = (V,
  E'_\phi)$ and $E^*_\phi$ \textcolor{myColor}{[See Lemma~\ref{lemma:queue_implies_active} in \ref{ap:lemma_queue_implies_active}]}. To extend a restricted Nash flow over time, we first compute a thin flow on $G'_\phi$ with resetting on $E^*_\phi$, and then extend the labels linearly as follows.
  For some $\alpha > 0$ we get for all $v \in V$, $e \in E$, $i = 1, \dots, \n$, and~$\varphi \in \:(\phi, \phi + \alpha]$ that
  \[\l_v(\varphi) \coloneqq \l_v(\phi) + (\varphi - \phi) \cdot \l'_v \quad \text{ and } \quad \begin{aligned}
  x_e(\varphi) &\coloneqq x_e(\phi) + (\varphi - \phi) \cdot x'_e, \\
  x_i(\varphi) &\coloneqq x_i(\phi) + (\varphi - \phi) \cdot x'_i. 
  \end{aligned}\]
  Based on this we can extend the inflow function and the inflow distribution, which gives us     
  \[f^+_e(\theta) \coloneqq \frac{x'_e}{\l'_u} \quad \text{for } \theta \in (\l_u(\phi), \l_u(\phi + \alpha)] \qquad \text{and} \qquad
  f_i(\varphi) \coloneqq x'_i \;=\; \l'_{\source_i} \cdot r_i \quad \text{for } \varphi \in (\phi, \phi + \alpha]\]
  for all $e = uv \in E$ and all $i = 1, \dots, \n$.
  Note that in the case of $\l'_u = 0$ the time interval is empty. Furthermore, it turns out that $f^-_e(\theta) = x'_e/\l'_v$ for all $\theta \in (\l_v(\phi), \l_v(\phi + \alpha)]$. \textcolor{myColor}{[See Lemma~\ref{lemma:outflow_of_extension} in \ref{ap:lemma_outflow_of_extension}.]}
  This extended flow over time together with the extended inflow distribution is called \emph{$\alpha$-extension} and it extends the Nash flow over time as long as the $\alpha$
  stays within the the following bounds:
  \begin{align}
  \l_v(\phi) - \l_u(\phi) + \alpha (\l'_v - \l'_u) &\geq \tau_e  \qquad \text{for all } e = uv \in E^*
  \label{eqn:alpha_resetting}\\
  \l_v(\phi) - \l_u(\phi) + \alpha (\l'_v - \l'_u) &\leq \tau_e  \qquad \text{for all } e = uv \in E \backslash E'.
  \label{eqn:alpha_others} 
  \end{align}
  The first inequality ensures that no flow can traverse an arc faster than its transit time. It holds with equality when
  the queue of $e$ vanishes at time~$\l_u(\phi + \alpha)$. The second inequality makes sure that all non-active arcs are unattractive for all particles in~$[\phi, \phi + \alpha)$. When it holds with
  equality the arc~$e$ becomes active for $\phi + \alpha$. When such an event occurs we must compute a new thin flow with resetting because either a resetting arc has become non-resetting or a
  non-active arc has become active. It is easy to see that there exists an $\alpha > 0$ that satisfies these
  inequalities since $\l_v(\phi) > \l_u(\phi) + \tau_e$ for arcs $e \in
  E^*_\phi$ and $\l_v(\phi) < \l_u(\phi) + \tau_e$ for arcs~$e \not \in E'_\phi$. \textcolor{myColor}{[See Lemma~\ref{lemma:queue_implies_active} in \ref{ap:lemma_queue_implies_active}.]}

    \begin{lemma} \label{lemma:extension_is_flow}
    The $\alpha$-extension forms a flow over time and the extended $\l$-labels coincide with the earliest arrival times,
    i.e., satisfy the Bellman equations~\eqref{eqn:bellman} for all $\varphi \in (\phi, \phi + \alpha]$.
    \end{lemma}
  The flow conservation follows immediately from the flow conservation of $x'$ and the Bellman equations are shown by
  distinguishing three cases. If the arc is non-active it stays non-active during the extended interval. For active, but
  non-resetting arcs that do not build up a queue, we obtain $\l_v(\phi + \xi) \leq \T_e(\l_u(\phi + \xi))$ from
  \eqref{eqn:l'_v_min} with equality if $\l'_v = \rho_e(\l'_u, x'_e)$. The same is true for resetting arcs or arcs that build up a queue, even though, the proof is a bit
  more technical. \textcolor{myColor}{[See \ref{ap:lemma_extension_is_flow} for a detailed proof.]}

  \begin{theorem} \label{thm:extension}
  Given a restricted Nash flow over time $((f^+_e)_{e\in E}, (f_i)_{i = 1}^\n)$ on $[0, \phi]$ and $\alpha > 0$    
  satisfying~\eqref{eqn:alpha_resetting} and~\eqref{eqn:alpha_others}, the $\alpha$-extension is a restricted
  Nash flow over time on~$[0, \phi + \alpha]$.
  \end{theorem}
  \begin{proof}
  We have $\sum_{i = 1}^\n f_i(\theta) =  \sum_{i = 1}^\n x'_i = 1$ for all $\theta \in \:(\phi, \phi + \alpha]$, which shows that $(f_i)_{i = 1}^n$ is a restricted inflow distribution. Lemma~\ref{lemma:nash_flow_characterization} yields
  $F_e^+(\l_u(\varphi)) = F_e^-(\l_v(\varphi))$ and $F_i(\varphi) = \l_{\source_i}(\varphi) \cdot r_i$ for
  all~$\varphi \in [0, \phi]$, so for $\xi \in (0, \alpha]$ it holds that
   \begin{align*}
   \!\!\!F_e^+(\l_u(\phi + \xi)) &= F_e^+(\l_u(\phi)) + x'_e / \l'_u \cdot \xi \cdot \l'_u = 
   F_e^-(\l_v(\phi)) + x'_e / \l'_v \cdot \xi \cdot \l'_v = F_e^-(\l_v(\phi + \xi)),\\
   F_i(\phi + \xi) &= F_i(\phi) + \xi \cdot x'_i = \l_{\source_i}(\phi) \cdot r_i + \xi \cdot \l'_{\source_i} \cdot r_i 
   = \l_{\source_i}(\phi + \xi) \cdot r_i.
  \end{align*}
  Again with Lemma~\ref{lemma:nash_flow_characterization} we have that the $\alpha$-extension is a restricted Nash
  flow on~$[0, \phi + \alpha]$. 
  \end{proof}
  Finally, we show that this construction leads to a Nash flow over time.
  \begin{theorem} \label{thm:finishing_nash_flow}
  There exists a Nash flow over time with multiple sources and a single sink.
  \end{theorem}
  In every iteration we find a positive $\alpha$ to extend the restricted Nash flow over time. If this series has a finite
  limit it is possible to compute the function values of the limit point and extend from there. In this manner it is
  possible to show that there exists a Nash flow over time including all particles by additionally proving that the
  cumulative inflow functions have to be unbounded. \textcolor{myColor}{[A detailed proof is given in \ref{ap:thm_finishing_nash_flow}.]}
    
\section{Multiple Sinks with Demands} \label{sec:demands} In this section we consider a graph $G = (V,E)$ as before except that it can have multiple sinks~$\Sinks \coloneqq
\set{\sink_1, \dots, \sink_\m}$ and demands~$d_1, \dots, d_\m > 0$ with $d_1 + \dots + d_\m = 1$. We show how to construct a Nash
flow over time in $G$ where a share of $d_j$ of the flow has $\sink_j$ as destination.

  \subparagraph*{Sub-flow over time decomposition} 
  In the following we define a sub-flow over time, which is, intuitively, a colored proportion of a
  flow over time satisfying flow conservation. Given a flow over time $f = (f^+_e)_{e\in E}$ with queue functions~$(z_e)_{e\in E}$, we consider a family of locally integrable and bounded inflow functions $g = (g^+_e)_{e \in E}$ with
  $g^+_e(\theta) \leq f^+_e(\theta)$ for almost all~$\theta \in [0, \infty)$.
  The corresponding outflow functions are obtained by the following consideration. For a point in time
  $\vartheta \in [0, \infty)$ let $\T_e^{-1}(\vartheta)$ be all times at which a particle could enter $e$ in
  order to leave it at time~$\vartheta$. Whenever $\T_e^{-1}(\vartheta)$ is not a singleton it is a proper
  interval and by \eqref{eqn:derivative_of_lambda} we have that $f^+_e(\theta) = 0$ for almost all~$\theta \in \T_e^{-1}(\vartheta)$.
  The \emph{sub-outflow function} for arc $e \in E$ is defined~as
  \begin{equation} \label{eqn:subflow_outflow}
  g^-_e (\vartheta) \coloneqq \begin{cases}
  f^-_e (\vartheta) \cdot \frac{g^+_e(\theta)}{f^+_e(\theta)} & \text{ if } f^+_e(\theta) > 0 \text{ and }
  \T_e^{-1}(\vartheta) = \Set{\theta},\\ 0 & \text{ else.}
  \end{cases}\end{equation}
  In other words, if $g^+_e(\theta)/f^+_e(\theta) \in [0, 1]$ is the inflow share of $g$ at time $\theta$, then the outflow share of $g$ has the same value at time~$\T_e(\theta)$.
  We call $g = (g^+_e)_{e\in E}$ a \emph{sub-flow over time of $f$} if for every $v \in V\backslash \Sources$ and almost all
  $\theta \in [0, \infty)$ we have
    \begin{equation} \sum_{e\in\delta^-(v)} g^-_e(\theta) - \sum_{e\in\delta^+(v)} g^+_e(\theta) \leq
    \sum_{e\in\delta^-(v)} f^-_e(\theta) - \sum_{e\in\delta^+(v)} f^+_e(\theta).
    \label{equ:subflow_conservation}\end{equation}
    Intuitively, this means that at every non-source node~$v$ the sub-flow over time $g$ can at most ``lose'' as much flow as~$f$ does.    
Furthermore, we say $g$ \emph{conserves flow at node $v \in
  V\backslash \Sources$} if   
    $\sum_{e\in\delta^-(v)} g^-_e(\theta) - \sum_{e\in\delta^+(v)} g^+_e(\theta) = 0$ holds for almost all $\theta \in [0, \infty)$.
    Note that if $f$ conserves flow at some node $v$, then $g$ does so as well.
We say $g$ is an \emph{$\Sources$-$\sink_j$-sub-flow over time} if it conserves flow at all nodes in $V \backslash \Set{\sink_j}$.
    
  Given an inflow distribution $(f_i)_{i=1}^n$ and a number $\gamma \in [0, 1]$, a family of locally integrable functions
  $(g_i)_{i=1}^\n$ with $g_i(\phi) \leq f_i(\phi)$ is called \emph{sub-inflow distribution of value~$\gamma$} if we have $\sum_{i = 1}^\n g_i(\phi) = \gamma$ for almost all~$\phi \in \Flow$. 
  To ensure that sub-flow is
  conserved at the sources we require the net flow leaving a source $\source_i$ at time $\T_i(\phi)$ to be equal
  to the amount of flow distributed to $\source_i$ at time $\T_i(\phi)$, which is
  $r_i \cdot g_i(\phi)/f_i(\phi) = g_i(\phi)/\T'_i(\phi)$,
  whenever $f_i(\phi) > 0$ and $0$ otherwise. More precisely, we say a sub-inflow distribution \emph{matches} a sub-flow over time if for almost all $\phi \in \Flow$ and all $i = 1, \dots, \n$ we have
  \[\T'_i(\phi) \cdot \left(\sum_{e \in \delta^+(\source_i)} g_e^{+}( \T_i(\phi)) - \sum_{e \in \delta^-(\source_i)}
  g_e^{-}(\T_i(\phi))\right) =  g_i(\phi).\]
  In this case we also say that the sub-flow over time \emph{conserves flow} at~$\source_i$ and that the sub-flow over time $g$ \emph{has value $\gamma$}.
  
  \begin{definition}[Sub-flow over time decomposition]
  A family of sub-flows over time $(g_e^{j+})_{e\in E}$ and matching sub-inflow distributions $(g_i^j)_{i = 1}^{\n}$ of value $\gamma_j$, for $j = 1, \dots, \m$, is called a \emph{sub-flow over time decomposition of $f$
  with values $\gamma_1, \dots, \gamma_\m$} if $\sum_{j=1}^\m \gamma_j = 1$ and
    \[g^{1+}_e(\theta) + \dots + g^{\m+}_e(\theta) = f^+_e(\theta) \quad \text{ for all }
 e \in E \text{ and almost all }\theta \in [0, \infty).\]
    \end{definition}
    Note that \eqref{eqn:subflow_outflow} implies 
    $\sum_{j = 1}^\m g^{j-}_e(\xi) = f^-_e(\xi)$ for all $e \in E$ and almost all $\xi \in [0, \infty)$.
  
  \subparagraph*{Nash flows over time with multiple sinks and demands} These sub-flow over time decompositions allow us to formalize Nash flows over time in the setting of multiple sinks with demands. Note that for the sake of clarity we omit the $+$ and simply write $f_e$ and $g_e^j$ for the inflow functions for the remaining of this paper. 
  \begin{definition}[Nash flow over time with demands]\label{defi:Nash_flows_with_demands}
    A tuple $f = ((f_e)_{e\in E}, (f_i)_{i = 1}^\n)$ consisting of a flow over time and an inflow distribution is a
    \emph{Nash flow over time with demands $d_1, \dots, d_\m$} if it satisfies the Nash flow conditions
    (\ref{eqn:Nash_condition_source}) and (\ref{eqn:Nash_condition_active}) from Definition~\ref{defi:Nash_flow} and, furthermore,
    has a sub-flow over time decomposition $((g_e^j)_{e\in E}, (g_i^j)_{i = 1}^{\n})_{j=1}^m$, such that $(g_e^j)_{e\in
    E}$ is an $\Sources$-$\sink_j$-sub-flow over time of value $d_j$ for all~$j = 1, \dots, \m$.
  \end{definition}
  
  To construct a Nash flow over time with demands we add a super sink~$\sink$ to the graph and use a single-sink Nash flow over time as constructed in Section~\ref{sec:dynamic_routing_game}.
  For this let~$\nu_{\min} \coloneqq \min_{e\in E} \nu_e$
   and $r_{\min} \coloneqq \min_{i = 1, \dots, \n} r_i$ be the minimal capacity/inflow rate and 
  $\sigma \coloneqq \min \Set{\nu_{\min}, r_{\min}}$.
  For all $j = 1, \dots, \m$ we define $\delta_j \coloneqq \min_{\source \in \Sources} d(\source, \sink_j)$, where $d(\source, \sink_j)$ is the length of a shortest $\source$\nobreakdash-$\sink_j$\nobreakdash-path according to the transit times. Furthermore, let $\delta_{\max} \coloneqq \max_{j = 1, \dots, \m} \delta_j$ be the maximal distance to some sink $\sink_j$ from its nearest source.
  We extend $G$ by a super sink~$\sink$ and $\m$ new arcs $e_j \coloneqq (\sink_j, \sink)$ with 
  \begin{equation}\label{equ:definition_of_nu}
  \tau_{e_j} \coloneqq \delta_{\max} - \delta_j \quad \text{ and } \quad \nu_{e_j} \coloneqq 1/2 \cdot d_j \cdot \sigma.
  \end{equation}
  We denote the \emph{extended graph} by $\exG \coloneqq (\exV, \exE)$ with $\exV \coloneqq V \cup \Set{\sink}$ and $\exE \coloneqq E \cup \Set{e_1,
  \dots, e_\m}$.
    
  Note, that the new capacities are strictly smaller than all original capacities and all inflow rates and that they are proportional to the demands. Furthermore, we choose the transit times such that all new arcs are in the current shortest paths network for particle $\phi = 0$. The reason for the choice of $\sigma$ is that for every thin flow with resetting $(x',\l')$ in $G$ we have $\l'_v \leq 1/\sigma$ for all $v \in V$. \textcolor{myColor}{[See Lemma~\ref{lemma:bound_for_l'} in \ref{ap:lemma_bound_for_l'}.]}

  We obtain a Nash flow over time with demands~$f$ by using a single-sink Nash flow over time $\exf$ in $\exG$, which exists
  due to Theorem~\ref{thm:finishing_nash_flow}. To prove this we first show that if all new arcs are active for some
  particle $\phi$ then there is a static flow decomposition of the thin flow with resetting~$x'$ with $x'_{e_j} = d_j$.
  This is formalized in the following lemma, where we write $x'\big|_E$ for the restriction of $x'$ to the original
  graph $G$ and $\abs{\,\cdot\,}$ for the flow value of a static flow.
  \begin{lemma} \label{lemma:thin_flow_decomp}
    Consider a thin flow with resetting $(x', \l')$ in $\exG$ where $\set{e_1, \dots, e_{\m}} \subseteq E'$, then there
    exists a static flow decomposition $x'\big|_E = x'^1 + \dots + x'^\m$ such that each static flow~$x'^j$ conserves flow
    on all~$v \in V\backslash (\Sources \cup \Set{\sink_j})$ and $\abs{x'^j} = d_j$ for~$j = 1, \dots, \m$.
  \end{lemma}
  The first part of this lemma can be shown with the well-known flow decomposition theorem and the statement that the
  flow values coincide with the demands follows from \eqref{eqn:l'_v_min} and \eqref{eqn:l'_v_tight} together with the
  fact that the arcs~$e_1, \dots, e_{\m}$ form an $\Sources$-$\sink$-cut. \textcolor{myColor}{[See
  \ref{ap:lemma_thin_flow_decomp}.]} In order to apply the previous lemma to all particles we show that the conditions
  are always met.
\begin{lemma} \label{lemma:new_arc_are_active}
In a Nash flow over time $\exf$ in $\exG$ the new arcs $e_1, \dots, e_{\m}$ are active for all particles $\phi \in \Flow$.
\end{lemma}
The key proof idea is the following. At the beginning $e_1, \dots, e_{\m}$ are active by the choice of their transit times. Later on the queues of all new arcs always increase since the capacities are sufficiently small. In general, a positive queue on an arc implies that it is active. \textcolor{myColor}{[See~\ref{ap:lemma_new_arc_are_active}.]}

By means of the previous lemmas we can finally prove that the Nash flow over time in $\exG$ induces a Nash flow over
time with demands in $G$.  \textcolor{myColor}{[See \ref{ap:thm_multi_sink}.]}
\begin{theorem} \label{thm:multi_sink}
  Let $\exf$ be a $\Sources$-$\sink$-Nash flow over time in~$\exG$. The flow over time $f \coloneqq \exf \big|_E$ on the
  original network together with the inflow distribution of $\exf$ is a Nash flow over time with demands~$d_1, \dots,
  d_\m$.
\end{theorem}
The sub-flow over time decomposition is obtained applying Lemma~\ref{lemma:thin_flow_decomp} and defining
\[g^j_e(\theta) \coloneqq \frac{x'^j_e}{\l'_u} \qquad \text{ and } \qquad
    g^j_i(\phi) \coloneqq \sum_{e\in \delta^+(\source_i)} x'^j_e - \sum_{e \in \delta^-(\source_i)} x'^j_e\]
   in every thin flow phase. That this forms a sub-flow over time decomposition, and therefore $f$ is a Nash
   flow over time with demands $d_1, \dots, d_{\m}$ can then be shown straightforwardly.
   
\section{Conclusion and Outlook} We showed that the Nash flow over time introduced in \cite{Koch_2009} can be extended
to our multi-terminal setting, for which we uncoupled the flow particles from their entering times and introduced inflow
distributions instead. Furthermore, the proper definition of a sub-flow-structure and a super-sink-construction allowed
us to have Nash flows over time with multiple sinks and demands. Nonetheless the much more challenging question about the
structure of a dynamic equilibrium in a setting with multiple origin-destination-pairs remains open. There are also
further interesting aspects that are unsolved in the original setting, such as the computational complexity of thin flows
with resetting or the question if the number of thin flow phases is finite within a Nash flow over time. Last but not
least the very interesting question if the price of anarchy is bounded or not remains open, despite some promising progress in recent time.
   
\newpage

  \bibliography{literature}

\begin{thebibliography}{10}

\bibitem{Cominetti_2011}
R.~Cominetti, J.~Correa, and O.~Larr{\'{e}}.
\newblock Existence and uniqueness of equilibria for flows over time.
\newblock In L.~Aceto, M.~Henzinger, and J.~Sgall, editors, {\em Automata,
  Languages and Programming}, volume 6756 of {\em Lecture Notes in Computer
  Science}, pages 552--563. Springer Berlin Heidelberg, 2011.
\newblock \href {https://doi.org/10.1007/978-3-642-22012-8_44}
  {\path{doi:10.1007/978-3-642-22012-8_44}}.

\bibitem{Cominetti_2015}
R.~Cominetti, J.~Correa, and O.~Larr{\'{e}}.
\newblock Dynamic equilibria in fluid queueing networks.
\newblock {\em Operations Research}, 63:21--34, 2015.
\newblock \href {https://doi.org/10.1287/opre.2015.1348}
  {\path{doi:10.1287/opre.2015.1348}}.

\bibitem{Cominetti2017}
R.~Cominetti, J.~Correa, and N.~Olver.
\newblock Long term behavior of dynamic equilibria in fluid queuing networks.
\newblock In F.~Eisenbrand and J.~K\"onemann, editors, {\em Integer Programming
  and Combinatorial Optimization}, volume 10328 of {\em Lecture Notes in
  Computer Science}, pages 161--172. Springer, 2017.
\newblock \href {https://doi.org/10.1007/978-3-319-59250-3_14}
  {\path{doi:10.1007/978-3-319-59250-3_14}}.

\bibitem{FordFulkerson58}
L.~R. Ford and D.~R. Fulkerson.
\newblock Constructing maximal dynamic flows from static flows.
\newblock {\em Operations Research}, 6:419--433, 1958.
\newblock \href {https://doi.org/10.1287/opre.6.3.419}
  {\path{doi:10.1287/opre.6.3.419}}.

\bibitem{FordFulkerson62}
L.~R. Ford and D.~R. Fulkerson.
\newblock {\em Flows in Networks}.
\newblock Princeton University Press, 1962.

\bibitem{HallHipplerSkut-TCS}
A.~Hall, S.~Hippler, and M.~Skutella.
\newblock Multicommodity flows over time: Efficient algorithms and complexity.
\newblock {\em Theoretical Computer Science}, 379:387--404, 2007.
\newblock \href {https://doi.org/10.1016/j.tcs.2007.02.046}
  {\path{doi:10.1016/j.tcs.2007.02.046}}.

\bibitem{Hoppe95}
B.~Hoppe.
\newblock {\em Efficient dynamic network flow algorithms}.
\newblock PhD thesis, Cornell University, 1995.

\bibitem{HoppeTardos00}
B.~Hoppe and \'{E}. Tardos.
\newblock The quickest transshipment problem.
\newblock {\em Mathematics of Operations Research}, 25:36--62, 2000.
\newblock \href {https://doi.org/10.1287/moor.25.1.36.15211}
  {\path{doi:10.1287/moor.25.1.36.15211}}.

\bibitem{Horni2016}
A.~Horni, K.~Nagel, and K.~Axhausen, editors.
\newblock {\em The Multi-Agent Transport Simulation MATSim}.
\newblock Ubiquity Press, London, 2016.
\newblock \href {https://doi.org/10.5334/baw} {\path{doi:10.5334/baw}}.

\bibitem{Kakutani1941}
S.~Kakutani.
\newblock A generalization of brouwer's fixed point theorem.
\newblock {\em Duke Math. J.}, 8:457--459, 1941.
\newblock \href {https://doi.org/10.1215/S0012-7094-41-00838-4}
  {\path{doi:10.1215/S0012-7094-41-00838-4}}.

\bibitem{Koch_2009}
R.~Koch and M.~Skutella.
\newblock Nash equilibria and the price of anarchy for flows over time.
\newblock {\em Theory of Computing Systems}, 49:323--334, 2009.
\newblock \href {https://doi.org/10.1007/978-3-642-04645-2_29}
  {\path{doi:10.1007/978-3-642-04645-2_29}}.

\bibitem{Rademacher1919}
H.~Rademacher.
\newblock {\"U}ber partielle und totale {D}ifferenzierbarkeit von {F}unktionen
  mehrerer {V}ariabeln und {\"u}ber die {T}ransformation der {D}oppelintegrale.
\newblock {\em Mathematische Annalen}, 79(4):340--359, 1919.
\newblock \href {https://doi.org/10.1007/BF01498415}
  {\path{doi:10.1007/BF01498415}}.

\bibitem{RanBoyce96}
B.~Ran and D.~E. Boyce.
\newblock {\em Modelling Dynamic Transportation Networks}.
\newblock Springer, Berlin, 1996.
\newblock \href {https://doi.org/10.1007/978-3-642-80230-0}
  {\path{doi:10.1007/978-3-642-80230-0}}.

\bibitem{Roughgarden05}
T.~Roughgarden.
\newblock {\em Selfish Routing and the Price of Anarchy}.
\newblock MIT Press, 2005.

\bibitem{SchloeterSkutella2017}
M.~Schl\"oter and M.~Skutella.
\newblock Fast and memory-efficient algorithms for evacuation problems.
\newblock In P.~N. Klein, editor, {\em Proceedings of the 28th Annual
  {ACM--SIAM} Symposium on Discrete Algorithms}, pages 821--840. SIAM, 2017.
\newblock \href {https://doi.org/10.1137/1.9781611974782.52}
  {\path{doi:10.1137/1.9781611974782.52}}.

\bibitem{Skutella_2009}
M.~Skutella.
\newblock An introduction to network flows over time.
\newblock In {\em Research Trends in Combinatorial Optimization}, pages
  451--482. Springer, 2009.
\newblock \href {https://doi.org/10.1007/978-3-540-76796-1_21}
  {\path{doi:10.1007/978-3-540-76796-1_21}}.

\end{thebibliography}

\newpage
\appendix
  \section{Additional Lemmas} 
  \subsection{Cumulative flows and queues} \label{ap:lemma_in_out_flow}
   \begin{lemma} \label{lemma:in_out_flow}
     For a given arc $e = uv \in E$ the following is true for all times $\theta \geq 0$:
     \begin{enumerate}
       \item $z_e(\theta) = F_e^+(\theta-\tau_e) - F_e^-(\theta)$ \label{item:queue}
       \item $F_e^+(\theta) = F_e^-(\T_e(\theta))$ \label{item:cumulative_in_equals_out}
      \end{enumerate}
    \end{lemma}
  
\begin{proof}
  We split the interval of entrance times $[0, \theta]$ into three subsets 
  \begin{align*}    
    \Theta_0 &\coloneqq \Set{\vartheta \in [0, \theta] | z_e(\vartheta) > 0} \\ 
    \Theta_1 &\coloneqq \Set{\vartheta \in [0, \theta] | z_e(\vartheta) = 0 \text{ and } f^+(\vartheta-\tau_e) > \nu_e}\\
    \Theta_2 &\coloneqq \Set{\vartheta \in [0, \theta] | z_e(\vartheta) = 0 \text{ and } f^+(\vartheta-\tau_e) \leq \nu_e}.
  \end{align*}
  By evolution of the queues~\eqref{eqn:z'_definition} and the definition of the outflow function~\eqref{eqn:defi_outflow} we obtain
  \begin{align*}z_e(\theta)
  &= \int_{\Theta_0 \cup \Theta_1} f_e^+(\vartheta-\tau_e) - \nu_e \diff \vartheta + \int_{\Theta_2} 0 \diff \vartheta \\
  &= \int_{\Theta_0 \cup \Theta_1} f_e^+(\vartheta-\tau_e)- f_e^-(\vartheta) \diff \vartheta +
  \int_{\Theta_2} f_e^+(\vartheta-\tau_e)  -  \underbrace{f_e^+(\vartheta-\tau_e)}_{=f_e^-(\vartheta)} \diff \vartheta \\
  &=F_e^+(\theta-\tau_e) - F_e^-(\theta).
  \end{align*}
  This shows~\ref{item:queue}.
  Since the continuous function $z_e$ cannot decrease faster than with slope $-\nu_e$ we obtain that if $z_e(\theta +
  \tau_e) > 0$ then it is positive during the interval $[\theta + \tau_e, \theta + \tau_e + z_e(\theta) / \nu_e)$ and \eqref{eqn:defi_outflow} yields that $f_e^-(\vartheta) = \nu_e$ for all~$\vartheta \in [\theta + \tau_e,
  \T_e(\theta))$. We obtain
  \begin{align*}
  F_e^-(\T_e(\theta)) &= F_e^-(\theta + \tau_e) + \int_{\theta + \tau_e}^{\T_e(\theta)}
  \underbrace{f_e^-(\vartheta)}_{= \nu_e} \diff \vartheta \\
  &\stackrel{\ref{item:queue}}{=} F_e^+(\theta) \underbrace{- z_e(\theta + \tau_e) + \nu_e \cdot
  (\T_e(\theta) - \theta - \tau_e)}_{= 0} = F_e^+(\theta),
  \end{align*}
  which shows~\ref{item:cumulative_in_equals_out}.           
\end{proof}

  \subsection{Characterization of active and resetting arcs} \label{ap:lemma_queue_implies_active}
  This lemma shows, among other facts, that every arc with a positive queue has to be active.
    \begin{lemma} \label{lemma:queue_implies_active}
      Consider a Nash flow over time $f$ with earliest arrival times~$(\l_v)_{v\in V}$. For every particle $\phi \in \Flow$, the following statements are true:
      \begin{enumerate}
      \item $E^*_\phi \subseteq E'_\phi$ \label{it:resetting_subset_active}
      \item $E'_\phi = \Set{e = uv \in E | \l_v(\phi) \geq \l_u(\phi) + \tau_e}$ \label{it:active_characterization} 
      \item $E^*_\phi = \Set{e = uv \in E | \l_v(\phi) > \l_u(\phi) + \tau_e}$  \label{it:resetting_characterization}
      \item The graph $G'_\phi = (V, E'_\phi)$ is acyclic and every node is reachable by a source. \label{it:current_paths_are_acyclic}
      \end{enumerate}
    \end{lemma}

  \begin{proof} 
  Recall that $e \in E^*_\phi \quad \Leftrightarrow \quad z_e(\l_u(\phi) + \tau_e) > 0$.
  \begin{enumerate}
    \item Let $e = uv \in E^*_\phi$ and $\phi' \coloneqq \max \Set{\varphi \leq \phi| e \in E'_\varphi}$, which is
    well-defined since $F_e^+(\phi) > 0$ implies for a Nash flow over time that $e$ has been active for a set with
    positive measure within $[0, \phi]$. Since $f$ is a Nash flow over time we have $f_e^+(\theta) = 0$ for almost
    all~$\theta \in (\l_u(\phi'), \l_u(\phi)]$, which implies that the queue cannot increase between
    $\l_u(\phi')+\tau_e$ and $\l_u(\phi) + \tau_e$.  Hence, $z_e(\l_u(\phi) + \tau_e) > 0$ yields that $z_e(\theta) > 0$
    for all $\theta \in (\l_u(\phi')+\tau_e, \l_u(\phi) + \tau_e]$. It follows that the travel time $\T_e$ is constant
    within $(\l_u(\phi'), \l_u(\phi)]$ because~\eqref{eqn:derivative_of_lambda} shows that~$\T'_e(\theta) = 0$. Together with
    the fact that $e$ is active for $\phi'$ and $\l_v$ is increasing we obtain
    \[\l_v(\phi) \stackrel{\eqref{eqn:bellman}}{\leq} \T_e(\l_u(\phi)) = \T_e(\l_u(\phi')) = \l_v(\phi') \leq \l_v(\phi).\]    
    Hence $e \in E'_\phi$, which shows~$E^*_\phi \subseteq E'_\phi$.

    \item Suppose~$z_e(\l_u(\phi) + \tau_e) > 0$. Then by \ref{it:resetting_subset_active} we have $e \in E^*_\phi \subseteq E'_\phi$, and therefore~$\l_v(\phi) = \T_e(\l_u(\phi)) > \l_u(\phi) + \tau_e$. Hence, $e$ is contained in both sides.
    For $z_e(\l_u(\phi)+\tau_e) = 0$ we obtain
    \begin{align*}e \in E'_\phi \;\Leftrightarrow\; \l_v(\phi) = \T_e(\l_u(\phi)) \;&\Leftrightarrow\; 
    \l_v(\phi) = \l_u(\phi) + \tau_e + \underbrace{z_e(\l_u(\phi)+\tau_e)/\nu_e}_{= 0} \\
    \;&\Leftrightarrow\; 
    \l_v(\phi) = \l_u(\phi) + \tau_e.\end{align*}
    This proves~$(ii)$.
    
    \item In the case that $e \not \in E'_\phi$ the previous results \ref{it:resetting_subset_active} and \ref{it:active_characterization} imply that $e$ is neither on the left nor on the right hand side.
    In the case of $e \in E'_\phi$ we conclude from the definition of $E^*_\phi$ and \ref{it:resetting_subset_active} that
    \begin{align*}
    e \in E^*_\phi \;\Leftrightarrow\; z_e(\l_u(\phi)+\tau_e) > 0 \;&\Leftrightarrow\; \underbrace{\l_u(\phi) + \tau_e + \frac{z_e(\l_u(\phi)+\tau_e)}{\nu_e}}_{= \l_v(\phi)} > \l_u(\phi) + \tau_e \\
    \;&\Leftrightarrow\; \l_v(\phi) > \l_u(\phi) + \tau_e.
    \end{align*}
    This shows~$(iii)$.
    
    \item Suppose there is a directed cycle of active arcs in $G'_\phi$. Since the sum of all transit times in every cycle is
    positive, it follows that not all $\l$-labels on the cycle can have the same value. So there has to be at least one
    arc $e = uv$ on the cycle with $\l_u(\phi) > \l_v(\phi)$, and therefore $\T_e(\l_u(\phi)) > \l_v(\phi)$, which is a
    contradiction since $e$ is active. Hence, $G'_\phi$ is acyclic.
    By the definition of the earliest arrival times every non-source node has an incoming active arc. Starting at $v$
    going backwards these arc shows that every node is reachable by a source, which finishes the proof.
    \end{enumerate}
  \end{proof}
  
  \subsection{Differentiation rule for a minimum}
     \begin{lemma}\label{lem:diff_rule_for_min}
     Let $E$ be a finite set and for every $e \in E$ let $\T_e \colon \R_{\geq 0} \rightarrow \R$ be a function that is differentiable almost everywhere. If we set $\l(\theta) := \min_{e \in E} \T_e(\theta)$ for all $\theta \geq 0$ it follows that
     \begin{equation} \label{eqn:diff_rule_for_min}
     \l'(\theta) = \min_{e \in E'_\theta} \T_e'(\theta)
     \end{equation}
     for almost all $\theta \geq 0$ where $E'_\theta := \set{e \in E| g(\theta) = g_e(\theta)}$.
     \end{lemma}
     
     \begin{proof}
     Let $\theta \geq 0$ such that $\l$ and all $\T_e$, $e \in E$, are differentiable, which is almost everywhere.
     The functions $\T_e$ are continuous at $\theta$ which gives us for sufficiently small $\varepsilon > 0$ that $\l(\theta + \xi) = \min_{e \in E'_\theta} \T_e(\theta + \xi)$ for all $\xi \in [\theta, \theta + \varepsilon]$. Hence,
     \begin{align*}
     \l'(\theta) &= \lim_{\xi \downarrow 0} \frac{\T(\theta + \xi) - \T(\theta)}{\xi}\\
     &= \lim_{\xi \downarrow 0}  \min_{e \in E'_\theta} \frac{\T_e(\theta + \xi) - \T(\theta)}{\xi}\\
     &= \min_{e \in E'_\theta} \lim_{\xi \downarrow 0} \frac{\T_e(\theta + \xi) - \T_e(\theta)}{\xi}\\
     &= \min_{e \in E'_\theta} \T'_e(\theta).
     \end{align*}
     \end{proof}   
  
   \subsection{Extended outflow function} \label{ap:lemma_outflow_of_extension}
    \begin{lemma} \label{lemma:outflow_of_extension}
    Let $((f^+_e)_{e\in E}, (f_i)_{i = 1}^\n)$ be a restricted Nash flow over time on $[0, \phi]$ and let $\alpha > 0$
    satisfy~\eqref{eqn:alpha_resetting} and~\eqref{eqn:alpha_others}. Then the outflow functions of the
    $\alpha$-extension satisfy
    $f_e^-(\theta) = \frac{x'_e}{\l'_v}$
    for all $e = uv \in E$ and almost all $\theta \in \;(\l_v(\phi), \l_v(\phi + \alpha)]$ and $f_e^-(\theta) = 0$ for~$\theta > \l_v(\phi + \alpha)$.
    \end{lemma}
    \begin{proof}
    Note that throughout this proof $\l_v(\varphi)$ for $\varphi > \phi$ is not the earliest arrival time, but the linear
    extension $\l_v(\varphi) \coloneqq \l_v(\phi) + (\varphi - \phi) \cdot \l'_v$. Let $I \coloneqq (\phi, \phi + \alpha]$
    be the flow of interest and $I_v := (\l_v(\phi, \l_v(\phi + \alpha)]$ for all nodes $v$. The particles in $[0, \phi]$ do not interfere with the outflow function $f^+_v$ within $I_v$, since otherwise the restricted Nash flow over time would not have chosen the fastest direction. We divide the proof into three cases.
    
    \subparagraph*{Case 1:}~$x'_e = 0$.
    
    Since $f_e^+(\theta) = x'_e / \l'_u = 0$ for all $\theta \in I_u$ we have that $f_e^-(\theta) = 0
    =  x'_e / \l'_v$ for all~$\theta \in I_v$ and of course $f_e^-(\theta) = 0$ for $\theta > \l_u(\phi + \alpha)$.
    
    \subparagraph*{Case 2:} $x'_e > 0$, $e \not\in E^*$ and~$x'_e / \nu_e \leq \l'_u$.
    
    We know that $e$ is active during $I$ and that $f_e^+(\theta) = x'_e / \l'_u \leq \nu_e$ for~$\theta \in I_u$.
    Furthermore, there is no queue at the beginning and no queue is building up. Therefore, we have~$\l_v(\phi) = \l_u(\phi)
    + \tau_e$. The definition of thin flows with resetting provides $\l'_u = \l'_v$ and together with the definition of the
    extension we obtain
    \[\l_u(\phi + \alpha) + \tau_e = \l_u(\phi) + \alpha \cdot \l'_u + \tau_e= \l_v(\phi) + \alpha \cdot \l'_v = 
    \l_v(\phi + \alpha).\]
    Hence, the last flow entering $e$ at time $\l_u(\phi + \alpha)$ leaves the edge at time $\l_v(\phi + \alpha)$ and since the outflow rate at time $\theta \in I_v$ equals the inflow rate at time $\theta - \tau_e \in I_u$ we get
    \[f^-_e(\theta) = f^+_e(\theta - \tau_e) = \frac{x'_e}{\l'_u} = \frac{x'_e}{\l'_v}.\]
    Furthermore, no flow enters $e$ after $\l_u(\phi + \alpha)$, and therefore the outflow function is zero after~$\l_v(\phi + \alpha)$.
    
    \subparagraph*{Case 3:} $x'_e > 0$ and ($e \in E^*$ or $x'_e / \nu_e > \l'_u$).
    
    This means there is either a queue at the beginning and throughout the phase or there is no queue at the beginning but
    immediately after $\phi$ a queue will build up. In either case, $e$ is active for all particles in $I$ and~$\l'_v = x'_e / \nu_e$. The
    inflow rate is $f^+_e(\theta) = x'_e/\l'_u$ for all $I_u$, and therefore the amount of flow entering $e$
    during this interval is
    \[A \coloneqq x'_e/\l'_u \cdot (\l_u(\phi + \alpha) - \l_u(\phi)) = x'_e/\l'_u \cdot
     (\l_u(\phi) + \alpha \cdot \l'_u - \l_u(\phi)) = x'_e \cdot \alpha.\]
     Since $\phi$ leaves $e$ at time $\l_v(\phi)$ and $f_e^-$ operates at capacity rate the last particle $\phi + \alpha$
     leaves $e$ at time
    \[\l_v(\phi) + A / \nu_e = \l_v(\phi) + \alpha \cdot x'_e / \nu_e =  \l_v(\phi) + \alpha \cdot \l'_v = \l_v(\phi +
     \alpha).\]
    Therefore, we have for all $\theta \in I_v$ that
    \[f_e^-(\theta) = \nu_e = \frac{x'_e}{\l'_v}.\]
    Since no flow is entering after $\l_u(\phi + \alpha)$ and particle $\phi + \alpha$ leaves $e$ at time $\l_v(\phi +
    \alpha)$ the outflow function is zero afterwards. This completes the proof.
    
    \end{proof}
  
   \subsection{Bound on node labels of thin flows with resetting} \label{ap:lemma_bound_for_l'}
     \begin{lemma}  \label{lemma:bound_for_l'}
     For every thin flow with resetting $(x',\l')$ in $G$ we have $\l'_v \leq 1/\sigma$ for all $v \in V$.
     \end{lemma}
      \begin{proof}
      It holds that $\l'_{\source_i} = x'_i / r_i$ and for $v \in V \backslash \Sources$
        the $\l'_v$ labels are equal to $\l'_u$ or $x'_e/\nu_e$ for some incoming arc~$e = uv$. It follows that all $\l'$ labels  in the original graph are bounded
        from above by
        \[\max\Set{\left(\max_{i = 1, \dots, \n} x'_i / r_i\right), \left(\max_{e\in E}  x'_e / \nu_e\right)} 
        \leq \max\Set{ 1 / r_{\min}, 1 /\nu_{\min}} = 1/\sigma.\]
        Note that all $x'_i$ and~$x'_e$ are bounded by $1$ from above since the flow value of $x'$ is~$1$.
      \end{proof}

  \section{Proofs}       
  \subsection{Proof of Lemma~\ref{lemma:nash_flow_characterization}} \label{ap:lemma_nash_flow_characterization}
    \begin{proof}   
      ``$\Rightarrow$'': Let $\xi \in [0, \phi]$ be the particle of largest value with~$F_e^+(\l_u(\xi)) =
      F_e^-(\l_v(\phi))$. This $\xi$ exists because of the intermediate value theorem, together with the fact that $F_e^+
      \circ \l_u$ is continuous and the following inequality, which follows by the monotonicity of $F_e^-$ and Lemma~\ref{lemma:in_out_flow}:
      \[F_e^+(\l_u(0)) = 0 \quad \leq \quad F_e^-(\l_v(\phi))\quad \leq \quad F_e^-(\T_e(\l_u(\phi)))
      = F_e^+(\l_u(\phi)).\] 
      Note that the second inequality is true because of~$\l_v (\phi) \leq \T_e(\l_u(\phi))$.      
      In the case of $\xi = \phi$ we are done since~$F_e^-(\l_v(\phi)) = F_e^+(\l_u(\xi)) = F_e^+(\l_u(\phi))$.
      Suppose~$\xi < \phi$. For all particles $\varphi \in (\xi, \phi]$ we know that~$\T_e(\l_u(\varphi)) \not=
      \l_v(\phi)$ because, otherwise, we had with Lemma~\ref{lemma:in_out_flow}~\ref{item:cumulative_in_equals_out} that $F_e^+(\l_u(\varphi)) = F_e^-(\T_e (\l_u(\varphi))) = F_e^-(\l_v(\phi))$
      which would contradict the maximality of~$\xi$. Hence, $e$ is not active for particles in $(\xi, \phi]$ which
      implies $f^+_e(\theta) = 0$ for almost all $\theta \in \l_u((\xi, \phi]) = (\l_u(\xi),\l_u(\phi)]$ since $f$ is a
      Nash flow over time. This leads to
      \[F_e^+(\l_u(\phi)) - F_e^-(\l_v(\phi)) = F_e^+(\l_u(\phi)) - F_e^+(\l_u(\xi)) = 
      \int_{\l_u(\xi)}^{\l_u(\phi)} f_e^+(\vartheta) \diff \vartheta = 0,\]
      which finishes the first part.
      The second part follows directly from
      $\l_{\source_i}(\phi) = \T_i(\phi) = F_i(\phi) / r_i \quad$ for all $i = 1, \dots, \n$.
      
      ``$\Leftarrow$'': Given a particle $\phi$ and an arc $e = uv$ such that $e$ is not active for $\phi$,
      i.e.,~$\l_v(\phi) < \T_e(\l_u(\phi))$. The continuity of $\l_v$ and
      $\T_e \circ \l_u$ implies that there is an $\varepsilon > 0$ with $\l_v(\phi + \varepsilon) <
      \T_e(\l_u(\phi - \varepsilon))$ and $e$ is not active for all particles in~$[\phi - \varepsilon, \phi +
      \varepsilon]$.
      This, the fact that $f_e^+$ and~$f_e^-$ are non-negative, and Lemma~\ref{lemma:in_out_flow} gives us
      \begin{eqnarray*} 0 &\leq& \int_{\l_u(\phi-\varepsilon)}^{\l_u(\phi+\varepsilon)} f_e^+(\vartheta) \diff \vartheta \\
        &=&
        \int_{\T_e(\l_u(\phi-\varepsilon))}^{\T_e(\l_u(\phi+\varepsilon))} f_e^-(\vartheta) \diff \vartheta \\
        &\leq& \int_{\l_v(\phi + \varepsilon)}^{\T_e(\l_u(\phi+\varepsilon))} f_e^-(\vartheta) \diff \vartheta \\
        &=& F_e^-(\T_e(\l_u(\phi+\varepsilon))) - F_e^-(\l_v(\phi + \varepsilon)) \\
        &=& F_e^+(\l_u(\phi+\varepsilon)) - F_e^-(\l_v(\phi + \varepsilon)) \\
        &\stackrel{(ii)}{=}& 0.
      \end{eqnarray*}    
      Hence, $f_e^+(\theta) = 0$ for almost all~$\theta \in [\l_u(\phi - \varepsilon), \l_u(\phi + \varepsilon)]$. 
      In other words, for almost all $\theta \in [0, \infty)$ it holds that~$\theta \not\in \l_u(\Phi_e) \Rightarrow
      f_e^+(\theta) = 0$. This is true because for $\theta \geq \l_u(0)$ we find a particle $\phi$ with $\l_u(\phi) = \theta$, due to the fact, that $\l_u$ is continuous and unbounded, and for all $\theta < \l_u(0)$ we have
      $f_e^+(\theta) = 0$, since no flow can reach $u$ faster than~$\l_u(0)$.   
      Finally, we get $\l_{\source_i}(\phi) = \T_i(\phi)$ since $\l_{\source_i}(\phi) \cdot r_i =
      F_i(\phi) = \T_i(\phi) \cdot r_i$ for all~$i = 1, \dots, \n$.
      This shows that $f$ is a Nash flow over time, which finishes the proof.
    \end{proof}
    
  \subsection{Proof of Theorem~\ref{thm:existence_of_NTF}} \label{ap:thm_existence_of_NTF}
  \begin{proof}
  We consider the following compact, convex, and non-empty set
  \[A \coloneqq \Set{\left((x'_i)_{i=1}^\n, (x'_e)_{e\in E}\right) | x'_i \geq 0, \quad \sum_{i = 1}^\n x'_i = 1, \quad 
  (x'_e)_{e\in E} \in K(E', x'_1, \dots, x'_{\n})}\]
  
  and the set-valued map $\Gamma \colon A \to 2^A$ defined by
  \[x' \mapsto \Set{y' \in A | \begin{matrix}[rl] y'_i = 0& \text{ for all } i \in \Set{1, \dots, \n} \text{ with } 
  \l'_{\source_i} < x'_i / r_i, \\ 
  y'_e = 0& \text{ for all } e = uv \in E' \text{ with } \l'_v < \rho_e(\l'_u, x'_e)\end{matrix}}  \]
  
  where $(\l'_v)_{v \in V}$ are the node labels associated with $x'$ given by the following Bellman equations
   \begin{alignat*}{2}
   \l'_{\source_i} &= \min\left(\{\:x'_i/r_i\:\} \cup \Set{\rho_e(\l'_u, x'_e) | e = u\source_i\in E'}\right)
   &&\quad \text{ for } i = 1, \dots, \n \\ 
   \l'_v &= \min_{e = uv\in E'} \rho_e(\l'_u, x'_e) &&\quad \text{ for } v
   \in V\backslash \Sources
   \end{alignat*}
   
   which are uniquely defined due to the fact that $G'$ is acyclic.
  We use the following version of the Kakutani's fixed point theorem \cite{Kakutani1941}.
  
  \begin{theorem}[Kakutani's Fixed Point Theorem]
  Let $A$ be a compact, convex and non-empty subset of $\R^N$ and $\Gamma\colon A \to 2^A$, such that for every
  $x' \in A$ the image $\Gamma(x')$ is non-empty and convex. Suppose the set $\Set{(x', y') | y' \in \Gamma(x')}$ is
  closed. Then there is a fixed point $x'_*$ of $\Gamma$, i.e.,~$x'_* \in \Gamma(x'_*)$.
  \end{theorem}
  
  We show that all conditions are satisfied.
  \begin{itemize}
  \item The set $\Gamma(x')$ is non-empty, because if we consider exactly the sources with $\l'_{\source_i} = x'_i / r_i$ and
  the arcs $e = uv$ with $\l'_v = \rho_e(\l'_u, x'_e)$, then there has to be at least one path $P$ from such a source
  $\source_i$ to the sink $\sink$. If we set $y'_i = 1$ and $y_e = 1$ for all arcs $e$ on $P$ and every other value to $0$
  we obtain an element in~$\Gamma(x')$. 
  
  \item Clearly, $\Gamma(x')$ is convex since the sources and arcs that can be used for sending flow are fixed within
  the set, and no convex combination of two elements uses sources or arcs different from the ones of the
  original elements.
  
  \item In order to show that $\Set{(x', y') | y' \in \Gamma(x)}$ is closed let $(x^n, y^n)_{n\in \N}$ be a sequence
  within this set, i.e.,~$y^n \in \Gamma(x^n)$. Since both sequences, $(x^n)_{n\in \N}$ and $(y^n)_{n\in \N}$, are contained in the compact set $A$ they both have a
  limit $x^*$ and $y^*$ within~$A$. Let $(\l^n)_{n \in \N}$ be the sequence of associated node labels of $(x^n)$ and $\l^*$ the
  node label of~$x^*$. Note that the mapping $x' \mapsto \l'$ is continuous, and therefore it holds that~$\l^* = \lim_{n \to
  \infty} \l^n$.
  
  We prove that~$y^* \in \Gamma(x^*)$. Suppose there is an $i \in \Set{1, \dots, \n}$ with $y^*_i > 0$ and~$\l^*_{\source_i} <
  x^*_i / r_i$. Then there has to be an $n_0 \in \N$ with $y^n_i > 0$ and $\l^n_{\source_i} < x^n_i / r_i$ for all~$n \geq
  n_0$. But this is a contradiction to~$y^n \in \Gamma(x^n)$. Suppose there is an arc $e = uv \in E'$ with $y^*_e > 0$
  and~$\l^*_v < \rho_e(\l^*_u, x^*_e)$. But again since $\rho_e$ is continuous there has to be an $n_0 \in \N$ such that
  $y^n_e > 0$ and $\l^n_v < \rho_e(\l^n_u, x^n_e)$ for all~$n \geq n_0$. Hence, $\Set{(x', y') | y' \in \Gamma(x)}$ is
  closed.
  \end{itemize}
  
  Since all conditions for the Kakutani's fixed point theorem are satisfied, there has to be a fixed point~$x^*$
  of~$\Gamma$. Let $\l^*$ be the corresponding node labeling. We show that it satisfies the thin flow
  conditions~\eqref{eqn:l'_s} to~\eqref{eqn:l'_v_tight}. If we have $x^*_i > 0$, then $\l^*_{\source_i} = x^*_i / r_i$
  follows from $x^* \in \Gamma(x^*)$. But also if $x^*_i = 0$, it holds that $0 \leq \l^*_{\source_i} \leq x^*_i / r_i =
  0$, and therefore we have equality, which yields~\eqref{eqn:l'_s}. Conditions~\eqref{eqn:l'_s_min}
  and~\eqref{eqn:l'_v_min} are satisfied by the construction of~$\l^*$. Finally, for every arc $e = uv \in E'$ with
  $x^*_e > 0$ it holds that $\l^*_v = \rho_e(\l^*_u, x^*_e)$ since $x^* \in \Gamma(x^*)$, which shows
  condition~\eqref{eqn:l'_v_tight}. This shows that $x^*$ together with $\l^*$ forms a thin flow with resetting which
  completes the proof.
  \end{proof}

  \subsection{Proof of Theorem~\ref{thm:l'_equations}} \label{ap:thm_l'_equations}
  
  \begin{proof}
  In Lemma~\ref{lemma:queue_implies_active} we showed that $G'_\phi$ and $E^*_\phi$ satisfy the preconditions.
  Furthermore, we have $x'_i(\phi) = f_i(\phi) \geq 0$ for all $i = 1, \dots, \n$ and $\sum_{i = 1}^\n x'_i(\phi) = \sum_{i = 1}^\n f_i(\phi) = 1$ for almost all~$\phi \in \Flow$. It remains to show that the equations~\eqref{eqn:l'_s} to~\eqref{eqn:l'_v_tight} are satisfied for almost all particles.
  For this let $\phi$ be a particle such that for all $e = uv$ the derivatives of $x_e$, $\l_v$, and $\T_e \circ \l_u$ exist and $x_e'(\phi)  = f_e^+(\l_u(\phi)) \cdot \l'_u(\phi) = f_e^-(\l_v(\phi)) \cdot \l'_v(\phi)$, which is almost everywhere.
  From Lemma~\ref{lemma:nash_flow_characterization} follows \eqref{eqn:l'_s} directly.
    
   For~\eqref{eqn:l'_s_min} and~\eqref{eqn:l'_v_min} first note that since $z_e$ is Lipschitz continuous, so is $\T_e$. We thus obtain from~\eqref{eqn:z'_definition} that the derivative of $\T_e(\theta)$ is almost everywhere
      \begin{equation} \label{eqn:derivative_of_lambda}
        \T'_e(\theta) = \begin{cases}
        \frac{f_e^+(\theta)}{\nu_e} & \text{ if } z_e(\theta + \tau_e) > 0,\\
        \max\Set{\frac{f_e^+(\theta)}{\nu_e}, 1} & \text{ if } z_e(\theta + \tau_e) = 0.
        \end{cases}
      \end{equation}
  In the case of $z_e(\l_u(\phi)+\tau_e) > 0$ we have
  \[\frac{\diff}{\diff \phi} \T_e(\l_u(\phi)) = \T'_e(\l_u(\phi)) \cdot \l'_u(\phi)
    \stackrel{\eqref{eqn:derivative_of_lambda}}{=} \left(\frac{f_e^+(\l_u(\phi))}{\nu_e}\right) \cdot \l'_u(\phi)
   = \frac{x'_e(\phi)}{\nu_e}\] 
  and if $z_e(\l_u(\phi)+\tau_e) = 0$, it holds that
  \[\frac{\diff}{\diff \phi} \T_e(\l_u(\phi)) \stackrel{\eqref{eqn:derivative_of_lambda}}{=} \max\Set{\frac{f_e^+(\l_u(\phi))}{\nu_e},
      1}\cdot \l'_u(\phi) = \max\Set{\frac{x'_e(\phi)}{\nu_e}, \l'_u(\phi)}.\] 
  Since the first case is equivalent to $e \in E^*_\phi$ and the second to $e \in E'_\phi \backslash E^*_\phi$ we obtain
 \[\frac{\diff}{\diff \phi} \T_e(\l_u(\phi)) = \rho_e(\l'_u(\phi), x'_e(\phi)).\]
 This equality together with the Bellman equations~\eqref{eqn:bellman} and Lemma~\ref{lem:diff_rule_for_min}, the differentiation rule for a minimum, provides
  \[\l'_{\source_i}(\phi) \leq \min_{e = u\source_i \in E'_\phi} \rho_e(\l'_u(\phi), x'_e(\phi)) \quad\text{ and }\quad
  \l'_v(\phi) = \min_{e = uv \in E'_\phi} \rho_e(\l'_u(\phi), x'_e(\phi)).\]
  For \eqref{eqn:l'_v_tight} suppose $x'_e(\phi) = f_e^-(\l_v(\phi)) \cdot \l'_v(\phi) > 0$. With \eqref{eqn:defi_outflow} we obtain
  \begin{align*}
 \l'_v(\phi) = \frac{x'_e(\phi)}{f_e^-(\l_v(\phi))} &= \begin{dcases}
     \frac{x'_e(\phi)}{\min\Set{f^+_e(\l_u(\phi)),\nu_e}} &\text{ if } z_e(\l_u + \tau_e) = 0,\\
     \frac{x'_e(\phi)}{\nu_e} &\text{ else},
   \end{dcases}\\
   &=\begin{dcases}
        \max\Set{\l'_u, \frac{x'_e(\phi)}{\nu_e}} &\text{ if } e \in E'_\phi \backslash E^*_\phi,\\
        \frac{x'_e(\phi)}{\nu_e} &\text{ if } e \in E^*_\phi,
      \end{dcases}\\
      & = \rho_e(\l'_u(\phi), x'_e(\phi)).
  \end{align*}
  This shows that the derivatives $(x'_i(\phi))_{i =1}^\n$, $(x'_e(\phi))_{e\in E'_\phi}$, and $(\l'_v(\phi))_{v\in V}$ form a thin flow with resetting.
  \end{proof}
  
  \subsection{Proof of Lemma~\ref{lemma:extension_is_flow}}\label{ap:lemma_extension_is_flow}
  
  \begin{proof}
  In order to prove that the $\alpha$-extension forms a flow over time we have to show that the flow conservation is
  fulfilled at every $v \in V\backslash \set{\sink}$, which is true because for all $\theta \in (\l_v(\phi), \l_v(\phi +
  \alpha)]$ it holds that
  \begin{align*}\sum_{e\in \delta^+(v)} f_e^+(\theta) - \sum_{e \in \delta^-(v)} f_e^-(\theta) 
  &= \sum_{e\in \delta^+(v)} x'_e / \l'_v - \sum_{e \in \delta^-(v)} x'_e / \l'_v \\
  &= \begin{cases}
  0 & \text{ if } v \in V \backslash (\Sources \cup \set{\sink}) \\
  x'_i / \l'_v = r_i & \text{ if } v = \source_i \in \Sources.
  \end{cases}\end{align*}
  For $\theta > \l_v(\phi + \alpha)$ all functions as well as the inflow rates are zero, and therefore the flow conservation holds as well.
     
  For the second part we show that the Bellman equations \eqref{eqn:bellman} for the earliest arrival times hold.
  Given an arc $e = uv \in E$, we distinguish between three cases.
  
  \subparagraph*{Case 1:}~$e \in E\backslash E'_\phi$.
  
  Since $\alpha$ satisfies equation~\eqref{eqn:alpha_others} it is satisfied for all $\xi \in \:(0, \alpha]$ and hence,
  \[\l_v(\phi + \xi) = \l_v(\phi) + \xi \cdot \l'_v \stackrel{\eqref{eqn:alpha_others}}{\leq} \l_u(\phi) + 
  \xi \cdot \l'_u + \tau_e \leq \T_e(\l_u(\phi)  + \xi \cdot \l'_u) = \T_e(\l_u(\phi + \xi)).\]
  
  \subparagraph*{Case 2:} $e \in E'_\phi \backslash E^*_\phi$ and~$\l'_u \geq x'_e/\nu_e$.
  
  Since $e$ is active we have $\l_v(\phi) = \T_e(\l_u(\phi)) = \l_u(\phi) + \tau_e$ and \eqref{eqn:l'_v_min} implies
  $\l'_v \leq \l'_u$. There is no queue building up, which means $z_e(\l_u(\phi + \xi)+\tau_e) = 0$ for all $\xi \in (0, \alpha]$. Combining these yields
  \[\l_v(\phi + \xi) = \l_v(\phi) + \xi \cdot \l'_v \stackrel{\eqref{eqn:l'_v_min}}{\leq} \l_u(\phi) + \tau_e + \xi \cdot \l'_u = \l_u(\phi + \xi) + \tau_e= \T_e(\l_u(\phi + \xi)).\]
  
  \subparagraph*{Case 3:}  $e \in E^*_\phi$ or ($e \in E'_\phi$ and $\l'_u < x'_e/\nu_e$).
  
  Again, $e$ is active, which means~$\l_v(\phi) = \T_e(\l_u(\phi)) = \l_u(\phi) + \tau_e + z_e(\l_u(\phi)+\tau_e)/\nu_e$.
  Additionally, $e \in E^*_\phi$ or $x'_e/\l'_u \leq \nu_e$ together with the thin flow condition~\eqref{eqn:l'_v_min}
  implies~$\l'_v \leq x'_e / \nu_e$. Since $f_e^+(\l_u(\phi)) - \nu_e = x'_e/\l'_u - \nu_e > 0$, equation~\eqref{eqn:z'_definition} implies~$z'_e(\l_u(\phi)+\tau_e) = f_e^+(\l_u(\phi)) - \nu_e= x'_e / \l'_u - \nu_e  $. Rearranging
  gives us,~$x'_e / \nu_e = z_e'(\l_u(\phi)+\tau_e) \cdot \l'_u / \nu_e + \l'_u$. Hence, for all $\xi \in \:(0, \alpha]$ we obtain with \eqref{eqn:l'_v_min} that
  \begin{align*}\l_v(\phi + \xi) &= \l_v(\phi) + \xi \cdot \l'_v \\
  &\leq \l_v(\phi) + \xi \cdot x'_e / \nu_e\\  
  &= \l_u(\phi) + \tau_e + z_e(\l_u(\phi)+\tau_e)/\nu_e+ \xi \cdot (z'(\l_u(\phi)+\tau_e) \cdot \l'_u / \nu_e + \l'_u)\\
  &= \l_u(\phi + \xi) + \tau_e + z_e(\l_u(\phi) + \tau_e + \xi \cdot \l'_u)/\nu_e \\
  &= \T_e(\l_u(\phi + \xi)).
  \end{align*} 
  This shows that there is no arc with an exit time earlier than the earliest arrival time, and therefore the left hand
  side of the Bellman equations is always smaller or equal to the right hand side. It remains to show that the
  equations hold with equality.
  For a source $\source_i$ we have $x'_i = f_i(\phi) = r_i \cdot \T'_i(\phi)$, and therefore
    \[\l_{\source_i}(\phi + \xi) = \l_{\source_i}(\phi) + \xi \cdot \l'_{\source_i} = \T_i(\phi) + \xi \cdot x'_i / r_i
     = \T_i(\phi) + \xi \cdot \T'_i(\phi) = \T_i(\phi + \xi)\]
     for all $\xi \in (0, \alpha]$.
    Hence, entering the network at a specific source is always a fastest option to reach it.
  For every node $v \in V \backslash \Sources$ there is at least one arc $e \in E'$ with $\l'_v
  = \rho(\l'_u, x'_e)$ in the thin flow due to \eqref{eqn:l'_v_min}. No matter if this arc belongs to Case 2 or
  Case 3 the corresponding equation holds with equality, which shows for all $\xi \in (0, \alpha]$ that
    \[\l_v(\phi + \xi) = \min_{e = uv \in E} \T_e(\l_u(\phi + \xi)).\]
  This completes the proof.
\end{proof}  

\subsection{Proof of Theorem~\ref{thm:finishing_nash_flow}} \label{ap:thm_finishing_nash_flow}
\begin{proof}
  In the first part we show that these $\alpha$-extensions lead to a restricted Nash flow on $[0,\infty)$. In
  the second part we prove, that all cumulative source inflow functions are unbounded, which shows that we have, indeed, a Nash flow over time.
  
  The process starts with the empty flow over time and the zero flow distribution, which is a restricted Nash flow
  over time for~$[0, 0]$. By applying Theorem~\ref{thm:extension} iteratively and choosing $\alpha$ maximal according
  to~\eqref{eqn:alpha_resetting} and~\eqref{eqn:alpha_others}, we obtain a sequence of restricted Nash flows over time
  $f^i$ for $[0,\phi_i]$ for $i = 1, 2 , \dots$, where the sequence $(\phi_i)_{i = 1}^\infty$ is strictly increasing. In
  the case that this sequence has a finite limit, say $\phi_\infty$, we can define a restricted Nash flow over time $f^\infty$ for
  $[0, \phi_\infty]$ by using the point-wise limit of the $x$- and $\l$-labels, which exists due to monotonicity and
  Lipschitz continuity of these functions. Then the process can be restarted from this limit point.
  
  Let $\mathcal{P}_G$ be the set
  of all particles $\phi \in \Flow$ for which there exists a restricted Nash flow over time on $[0, \phi]$ constructed as
  described above. The set $\mathcal{P}_G$ cannot have a maximal element because this could be extended by using
  Theorem~\ref{thm:extension}. But it also cannot have an upper bound since the limit of any convergent sequence would be
  contained in this set. Therefore, there exists an unbounded increasing sequence $(\phi_i)_{i = 1}^\infty \in \mathcal{P}_G$. From the
  corresponding restricted Nash flows over time we can construct the restricted Nash flow over time $f$ on $[0, \infty)$ by taking the point-wise limit of the $x$- and $\l$-labels.
  
  It remains to show that the inflow distribution of this restricted Nash flow over time is unbounded. For this we first
  show that the earliest arrival time $\l_{\sink}$ is unbounded. There cannot be an upper bound $B$
  on~$\l_{\sink}$ since the flow rate into $\sink$ is bounded by $N \coloneqq \sum_{e \in\delta^-(\sink)} \nu_e$ and
  with the FIFO principle we obtain that no particle~$\phi > N \cdot B$ reaches $\sink$ before time~$\phi / N > B$.
  Next, we show that all $\l$-labels are unbounded. Suppose this is not true. Since every node can reach $\sink$ there
  would be an arc $e = uv$, where $\l_u$ is bounded and $\l_v$ is not. Since $T_e$ is Lipschitz continuous $T_e \circ \l_u$
  would be bounded as well. But this contradicts that $\l_v(\phi) \leq \T_e(\l_u(\phi))$ goes to infinity for $\phi \to \infty$.
  Hence, $F_i(\phi) = \l_{\source_i}(\phi) \cdot r_i$ is unbounded for every $i = 1, \dots, n$, which completes the proof.
  \end{proof}

 \subsection{Proof of Lemma~\ref{lemma:thin_flow_decomp}} \label{ap:lemma_thin_flow_decomp}
    \begin{proof} 
      Let $\Paths$ be the set of all $\Sources$-$\sink$-paths in the current shortest paths network~$G' = (V, E')$. 
      Note, that $G'$ is always acyclic and $x'$ can, therefore, be described by the path vector~$(x'_P)_{P\in \Paths}$ due to the well-known flow decomposition theorem.  
      For $j = 1, \dots, \m$ let $\Paths_j$ be the set of all $\Sources$-$\sink$-paths that contain~$e_j$. These sets form a partition of $\Paths$ since every path has to use exactly one of the new arcs. By setting $x'^j \coloneqq \sum_{P
      \in \Paths_j} x'_P\big|_E$ we obtain the desired decomposition of $x'$, because $x'_P\big|_E$ for $P\in \Paths_j$
      conserves flow on all nodes except the ones in~$\Sources \cup \Set{\sink_j}$ and the same is true for their
      sums.
    
      Since $x'^j$ sends $\abs{x'^j}$ flow units from $\Sources$ over $e_j$ to $\sink_j$ we have $\abs{x'^j} = x'_{e_j}$.
      It remains to show that $x'_{e_j} = d_j$ for all~$j = 1, \dots, \m$.
      Suppose this is not true. Since $x'$ sends exactly $d_1 + \dots + d_{\m} = 1$ flow
      units from $\Sources$ to $\sink$, there has to be an index $a \in \Set{1, \dots, \m}$ with $x'_{e_a} > d_a$ and an index $b \in \Set{1, \dots,
      \m}$ with~$x'_{e_b} < d_b$.
      
      With Lemma~\ref{lemma:bound_for_l'} it follows that
      \[\l'_{\sink_b} \leq \frac{1}{\sigma} \stackrel{\eqref{equ:definition_of_nu}}{<}
      \frac{d_a}{\nu_{e_a}} < \frac{x'_{e_a}}{\nu_{e_a}} \stackrel{\eqref{eqn:l'_v_tight}}{\leq} \l'_{\sink}
        \quad \text{ as well as } \quad
        \frac{x'_{e_b}}{\nu_{e_b}} \stackrel{\eqref{equ:definition_of_nu}}{=} \underbrace{\frac{x'_{e_b}}{d_b}}_{< 1} \cdot
        \frac{2}{\sigma} < \underbrace{\frac{x'_{e_a}}{d_a}}_{> 1} \cdot \frac{2}{\sigma}
        \stackrel{\eqref{equ:definition_of_nu}}{=} \frac{x'_{e_a}}{ \nu_{e_a}} \stackrel{\eqref{eqn:l'_v_tight}}{\leq} \l'_{\sink}.\]
      But this is a contradiction, because \eqref{eqn:l'_v_min} yields that $\l'_{\sink} = \min\limits_{j = 1, \dots, \m}
      \rho_{e_j}(\l'_{\sink_j}, x'_{e_j})$ and the last two equations show~$\rho_{e_b}(\l'_{\sink_b}, x'_{e_b}) <
      \l'_{\sink}$. Hence, we have $\abs{x'^j} = d_j$ for all~$j = 1, \dots, \m$, which finishes the proof.
    \end{proof}

\subsection{Proof of Lemma~\ref{lemma:new_arc_are_active}} \label{ap:lemma_new_arc_are_active}
\begin{proof}
    For particle $\phi = 0$ there are no queues yet, and therefore the exit time for each arc $e$ is~$\T_e(\theta) =
    \theta + \tau_e$. Hence, $\l_{\sink_j}(0) = \delta_j$ for all $j = 1, \dots, \m$ and by construction we
    have~$\l_{\sink}(0) = \l_{\sink_j}(0) + \tau_{e_j} = \T_{e_j}(\l_{\sink_j}(0))$ for~$j = 1,\dots,m$. Therefore, all arcs
    $e_j$ are active in the beginning and also during the first thin flow phase because by Lemma~\ref{lemma:thin_flow_decomp} we have $x'_{e_j} > 0$ for the first thin flow with resetting which implies that $e_j$ stays active.
    
    Suppose now for contradiction that there are particles for which not all new arcs are active. Let $\phi_0$ be the
    infimum of these particles. By the consideration above we have $\phi_0>0$ and Lemmas~\ref{lemma:bound_for_l'} and \ref{lemma:thin_flow_decomp}
    imply
    \[f^+_{e_j}(\l_{\sink_j}(\phi))=\frac{x'_{e_j}}{\l'_{\sink_j}}
     \geq x'_{e_j} \cdot \sigma = d_j \cdot \sigma \stackrel{\eqref{equ:definition_of_nu}}{>} \nu_{e_j}\]
     for almost all $\phi \in [0, \phi_0)$ and all $j=1,\dots,m$. Hence,~\eqref{eqn:z'_definition} yields
     $z'_{e_j}(\l_{\sink_j}(\phi) + \tau_{e_j}) = f^+_{e_j}(\l_{\sink_j}(\phi)) - \nu_{e_j} > 0$ and, together with the
     fact that $\l'_{t_j} > 0$ (due to the positive throughput of $x'$ at $t_j$), we obtain
    \[\frac{\diff}{\diff \phi}  z_{e_j}(\l_{\sink_j}(\phi) + \tau_{e_j}) = z'_{e_j}(\l_{\sink_j}(\phi) + \tau_{e_j}) \cdot \l'_{\sink_j} > 0.\]
    In other words, a queue is building up within $[0,\phi_0)$, and therefore $z_{e_j}(\l_{\sink_j}(\phi_0) + \tau_{e_j})
    > 0$ for all $j = 1,\dots, m$. But the continuity of $z_{e_j}$ implies that there will be positive queues for all
    $\phi \in [\phi_0, \phi_0 + \varepsilon]$ for sufficiently small $\varepsilon > 0$. By
    Lemma~\ref{lemma:queue_implies_active} this implies that all new arcs are active during this interval contradicting that
    $\phi_0$ is an infimum.
\end{proof}

\subsection{Proof of Theorem~\ref{thm:multi_sink}}\label{ap:thm_multi_sink}
\begin{proof}
  It remains to show that the thin flow decompositions of the particles in $\Flow$ correspond to a sub-flow over time
  decomposition of the Nash flow over time. Throughout this proof we denote $\delta^-(v)$ and $\delta^+(v)$ for the in-
  and out-going arcs of $v$ within the original network $G$. Let~$I\coloneqq [a, b)$ be an interval such that
  the thin flow with resetting is constant $(x', \l')$ for all particles in $I$. For every node $v$ we denote by $I_v
  \coloneqq [\l_v(a),\l_v(b))$ the interval of local times of particles in $I$. By Lemma~\ref{lemma:new_arc_are_active}
  all new arcs $e_1, \dots, e_{\m}$ are active. Let $x'^1, \dots, x'^m$ be the thin flow decomposition given by
  Lemma~\ref{lemma:thin_flow_decomp}. The corresponding decomposition for the Nash flow over time with demands is constructed
  by setting
    \begin{alignat*}{2}
    g^j_e(\theta) &\coloneqq \frac{x'^j_e}{\l'_u} \qquad &&\text{for } \theta \in I_u\\ 
    g^j_i(\phi) &\coloneqq \sum_{e\in \delta^+(\source_i)} x'^j_e - \sum_{e \in \delta^-(\source_i)} x'^j_e \qquad
    &&\text{for } \varphi \in I
    \end{alignat*}
  for all $j = 1, \dots, \m$, every $e = uv \in E$, and all $i = 1, \dots, \n$. Note that if $\l'_u = 0$ we have
  $\l_u(a) = \l_u(b)$, and therefore $I_u$ is empty. By setting $g_e^j(\theta) \coloneqq 0$ for all $\theta < \l_u(0)$ we
  obtain well-defined functions $g_e^j$.
  
  First, we show that $g^j$ satisfies the sub-flow over time properties and conserves flow at all nodes
  except~$\Sources \cup \Set{\sink_j}$ for all $\theta \in I_u$.
  
  Given an arc $e = uv$ we obviously have for all $\theta \in I_u$ that
  \[g^{j}_e(\theta) = x'^j_e/\l'_u \leq x'_e/\l'_u = f_e(\theta).\] 
  If $x'_e > 0$ we have $f_e(\theta) = x'_e/\l'_u > 0$ for almost all $\theta \in
  I_u$ and by the definition of $g^{j-}$ we get for almost all $\xi \in I_v = \T_e (I_u)$ and $\theta
  \in I_u$, the unique value with $\xi = \T_e(\l_u(\phi))$, that
  \[g^{j-}_e(\xi) = f^-_e(\xi) \cdot \frac{g^{j}_e(\theta)}{f_e(\theta)} = \frac{x'_e}{\l'_v} \cdot
  \frac{x'^j_e}{\l'_u} \cdot \frac{\l'_u}{x'_e} = \frac{x'^j_e}{\l'_v}.\]
  But this equality also holds if $x'_e = 0$ because in this case it holds that $f_e(\theta) = 0$ for almost all $\theta
  \in I_u$, and therefore we have by definition that~$g^{j-}_e(\xi) = 0$. The following equation shows that $g^j$
  conserves flow at all nodes $v \in V \backslash \Sources \cup \Set{\sink_j}$ for almost all~$\theta \in
  I_v$
  \[ \sum_{e\in\delta^-(v)} g^{j-}_e(\theta) - \sum_{e\in\delta^+(v)} g^{j}_e(\theta) = \sum_{e\in\delta^-(v)} 
  \frac{x'^j_e}{\l'_v} - \sum_{e\in\delta^+(v)} \frac{x'^j_e}{\l'_v} = \frac{1}{\l'_v} \cdot \left(
  \sum_{e\in\delta^-(v)}
  x'^j_e - \sum_{e\in\delta^+(v)} x'^j_e \right) = 0,\]
  where the last equality holds because of the flow conservation of $x'^j$ at~$v$. To show
  $\eqref{equ:subflow_conservation}$ it remains to prove it for $\sink_j$, which is true because for all $\theta \in
  I_{\sink_j}$ we have
  \[ \sum_{e\in\delta^-(\sink_j)} g^{j-}_e(\theta) - \sum_{e\in\delta^+(\sink_j)} g^{j}_e(\theta) 
  = \frac{x_{e_j}}{\l'_{\sink_j}} = \exf_{e_j}(\theta) 
  = \sum_{e\in\delta^-(\sink_j)} f^-_e(\theta) - \sum_{e\in\delta^+(\sink_j)} f_e(\theta).\]

  Next, we show that $(g_i^j)_{i = 1}^n$ is a matching sub-inflow distribution for all $j = 1, \dots, m$ with values
  $d_j$ for all~$\phi \in I$. In the case of $T'(\phi) \stackrel{\eqref{eqn:Nash_condition_source}}{=}\l'_{\source_i} >
  0$ it holds that
  \[\left(\sum_{e\in\delta^+(\source_i)}\!\! g^{j}_e(\l_{\source_i}(\phi)) - 
  \sum_{e\in\delta^-(\source_i)}\!\! g^{j-}_e(\l_{\source_i}(\phi))\right)
  \T'_i(\phi) = \left(\sum_{e\in\delta^+(\source_i)}\! \frac{x'^j_e}{\l'_{\source_i}} -
  \sum_{e\in\delta^-(\source_i)}\! \frac{x'^j_e}{\l'_{\source_i}}\right)
  \l'_{\source_i} = g_i^j(\phi).\]
  In the case of $\l'_{\source_i} = 0$ this is also true since both sides are equal to~$0$. By
  Lemma~\ref{lemma:thin_flow_decomp} we obtain for all~$\phi \in I$ that
    \[\sum_{i = 1}^\n g^j_i(\phi) = \sum_{i = 1}^\n \left(\sum_{e \in \delta^+(\source_i)} x'^j_e - 
    \sum_{e \in \delta^-(\source_i)} x'^j_e\right) = \abs{x'^j} = d_j.\]
  Finally, we show that the family $(g^j)_{j=1}^\m$ together with the matching sub-inflow distributions fulfills the
  sub-flow over time decomposition conditions for all $\theta \in I_u$. Clearly, $\sum_{j=1}^m d_j = 1$ and for all $e =
  uv \in E$ we have
  \[\sum_{j = 1}^\m g^j_e(\theta) = \sum_{j = 1}^\m \frac{x'^j_e}{\l'_u} = \frac{x'_e}{\l'_u} = f_e(\theta).\]
  Note that all these previous conditions hold for all $\phi \in \Flow$ and all $\theta \in [0, \infty)$ because either
  $\theta < \l_v(0)$, where all in and out flow at $v$ is $0$, or $\theta$ is element of the local times $I_u$ of some
  particle interval $I$. Hence, $(g^j)_{j=1^m}$ is a sub-flow over time decomposition of $f$ with values $d_1, \dots,
  d_{\m}$, where $g^j$ is an $\Sources$-$t_j$-sub-flow over time. Since $\exf$ is a Nash flow over time $f$ satisfies the
  Nash flow conditions \eqref{eqn:Nash_condition_source} and \eqref{eqn:Nash_condition_active} as well, and therefore $f$
  is a Nash flow over time with demands $d_1, \dots, d_{\m}$.
\end{proof}    
  
\end{document}